\documentclass[letterpaper, 10 pt, conference]{ieeeconf}  % Comment this line out if you need a4paper

\IEEEoverridecommandlockouts                              % This command is only needed if 
\usepackage{amsmath} % assumes amsmath package installed
\usepackage{amssymb}  % assumes amsmath package installed,

\usepackage{subcaption}
\usepackage{graphicx}

\usepackage{amsthm}
\theoremstyle{definition}
\newtheorem{proposition}{Proposition}
\theoremstyle{definition}

\theoremstyle{definition}

\theoremstyle{definition}
\newtheorem{remark}{Remark}
\theoremstyle{definition}

\theoremstyle{definition}
\newtheorem{theorem}{Theorem}
\newtheorem{assumption}{Assumption}

\usepackage{cite}
\usepackage{balance}

\usepackage{color}
\graphicspath{ {figs/} }

\title{\LARGE \bf
Decentralized temperature and storage volume control in multi-producer district heating*
}

\author{ Juan E. Machado$^{1}$, Joel Ferguson$^{2}$, Michele Cucuzzella$^{3,1}$ and Jacquelien M.~A.~Scherpen$^{1}$ % <-this % stops a space
\thanks{*This research received funding from the Dutch Research Council (NWO), ERA-Net Smart Energy Systems and 
European Union's Horizon 2020 research and innovation programme grant no.~775970.}% <-this % stops a space
\thanks{$^{1}$J. E. Machado, M. Cucuzzella  J. M. A. Scherpen are with the Jan C. Willems Center for Systems and Control, ENTEG, Faculty of Science and
Engineering, University of Groningen, Nijenborgh 4, 9747 AG Groningen, the
Netherlands {\tt\small j.e.machado.martinez(j.m.a.scherpen)@rug.nl}}%
\thanks{$^{2}$Joel Ferguson is with the School of Engineering, The University of Newcastle, Australia {\tt\small joel.ferguson@newcastle.edu.au}}%
\thanks{$^{3}$M. Cucuzzella is also with the Department of Electrical, Computer and Biomedical Engineering, University of Pavia, via Ferrata 5, 27100 Pavia, Italy {\tt \small michele.cucuzzella@unipv.it}}
}

\begin{document}

\maketitle
\thispagestyle{empty}
\pagestyle{empty}

%%%%%%%%%%%%%%%%%%%%%%%%%%%%%%%%%%%%%%%%%%%%%%%%%%%%%%%%%%%%%%%%%%%%%%%%%%%%%%%%
\begin{abstract}
Modern district heating technologies have a great potential to make the energy sector more flexible and sustainable due to their capabilities to use energy sources of varied nature and to efficiently store energy for subsequent use. Central control tasks within these systems for the efficient and safe distribution of heat refer to the stabilization of overall system temperatures and the regulation of storage units state of charge. These are challenging goals  when the networked and nonlinear nature of district heating system models is taken into consideration. In this letter, for district heating systems with multiple, distributed heat producers, we  propose a decentralized control scheme to provably meet said tasks stably. 
\end{abstract}

%%%%%%%%%%%%%%%%%%%%%%%%%%%%%%%%%%%%%%%%%%%%%%%%%%%%%%%%%%%%%%%%%%%%%%%%%%%%%%%%
\section{Introduction}

District heating systems (DHSs) distribute heat from heating stations, conventionally a single combined heat and power plant, towards clusters of consumers within a neighborhood or a small city using a network of underground insulated pipelines \cite{Lund2014}. The newest generation DHSs have the potential to improve the sustainability of a fossil fuel-dependent heating sector by increasing the share of distributed generations units based on renewables (e.g., geothermal or solar thermal) as well as units based on residual heat from some industrial processes. Such units can be flexibly incorporated with the support of heat storage devices (typically water tanks) \cite{Lund2014} (see also \cite{guelpa2019thermal,novitsky2020smarter}).

Heat distribution in DHSs, in the face of varying weather conditions and heat demand profiles, is supported by a control system in charge of regulating the supply (return) temperature of heat producers (consumers) and the state of charge of storage units \cite{vandermeulen_control_review_18,buffa2021advanced}. This is done by adjusting producers' power injections and the overall system's flow rates so that the heat supply matches the demand and the system temperatures stay within prescribed domains. 
Optimal, predictive (and centralized) control of DHSs has been addressed for systems without (\cite{sandou_predictive_05,krug_DH_2021}) and with (\cite{verrilli2016model}) heat storage units.
% {\color{red} \cite{sandou_predictive_05,verrilli2016model} use mostly static (algebraic) models  to describe most the system's components. Very detailed DH systems dynamics are considered in \cite{krug_DH_2021}, which are discretized for the purpose of formulating an optimal control problem.  In the three works closed-loop stability analysis is not performed. And the schemes are centralized.} 
 Supply temperature control and storage volume regulation is performed in \cite{Scholten_tcst_2015} on a dynamic,  nonlinear system model of a DH system comprising a single heating station with an adjacent storage tank. The control design is based on the internal model principle  and offers closed-loop stability guarantees and robustness against uncertainty of some parameters.
 %{\color{red} This work is restricted to DH systems with a single producer and single storage tank. The temperature dynamics of the distribution network are neglected. {\color{blue}They do consider time-varying consumer demand profiles, which are represented as the output of an exosystem.}} 
 For a similar system model, which does not consider storage units,  but does consider  transport dynamics of the distribution network (codified as a delay), a control design based on Lyapunov-Krasovksii theory to achieve (supply and return) temperature regulation is reported in  \cite{bendtsen_control_2017}. 
 %{\color{red} Focus is also on single producer case. It is not clear if the design can be extended to multiple-producer case. No storage tanks is considered. Distribution network is modeled using delays, however, delays are omitted in the return layer.}
Recently, the control of supply temperatures of heating stations  was  investigated in \cite{krishna_and_schiffer_21} using passivity-based design within the context of electro-thermal microgrids consisting of distributed generation units. 
%{\color{red} Temperature regulation of comsumers return temperature is not addressed, and even though storage tanks are considered,  volume control seems to be missing.} 
Optimal, open-loop end-user temperature control is investigated in \cite{Alisic2019} via  numerical simulations of a DH system with multiple {\em prosumers}. 
%{\color{red} Even though storage tanks are considered, they interact via heat exchangers, so no volume control is needed. In this work moreover, the distribution network is modeled via algebraic, not dynamic expressions. {\color{blue} They do have a more detailed end-user model in some sense, as they consider room temperature dynamics.}} 
Focused on the hydraulic dynamics, i.e., neglecting the thermal behavior, the works \cite{DePersis2011,DePersis2014,Scholten2017} address  pressure regulation on single producer DH systems and \cite{jmnj_2021_adaptive,Trip2019a} consider flow and volume control in the multi-producer setting.

%
%\subsection{State-of-the-art}
%
%\begin{enumerate}
%
%\item Flow/pressure control \cite{DePersis2011}, \cite{DePersis2014}, \cite{Scholten2017}, \cite{felix_4gdh_21}. Note: \cite{felix_4gdh_21} doesn't actually do control. But the modular modeling approach paves the way for both flow and pressure control using frameworks such as port-Hamiltonian systems theory.
%
%\item  Volume control \cite{Trip2019a}, \cite{Alisic2019}. Note: \cite{Alisic2019} does not consider volume control, but rather open-loop temperature control. However, both papers \cite{Trip2019a}, \cite{Alisic2019} consider prosumers.
%
%\item Temperature control \cite{sandou_predictive_2005}, \cite{Scholten_tcst_2015}, \cite{bendtsen_control_2017}, \cite{krishna_and_schiffer_21}.
%
%\item The group of V. Mehrmann has also publications about optimal control of district heating systems. I have also seen a couple of papers citing \cite{Hauschild2020} which might potentially do control also. This needs to be revised carefully. The relevant paper is  \cite{krug_DH_2021}. Here they do detailed modeling suitable for optimal control. 
%
%
%\end{enumerate}

\subsection{Contributions}

In Section~\ref{sec:model}, we describe the setup of the considered DH system, which features multiple, distributed producers with heat storage capabilities. We then develop a (modular) dynamic model that  describes the behavior of the overall system temperatures, including heat exchangers and storage tanks, and  the  supply and return layers of the distribution network.  Our main contribution is in Section~\ref{sec:control}, where we design novel decentralized controllers to regulate  supply temperature of heat producers, the temperature and volume of the storage tanks and  the return temperature of heat consumers. Stability analysis of the overall closed-loop system concludes the section.  Note that the simultaneous treatment of the mentioned  tasks, using decentralized controllers and with the considered system setup is not addressed in the above cited works. 
%{\color{red}Throughout the letter we discuss relevant similarities, differences and limitations of our results with respect to said references.}

%{\color{red}
%Keep in mind the following to motivate our contributions and what we could eventually address with them:
%\begin{itemize}
%
%\item Control is  challenging  when is taken into consideration the networked and nonlinear nature of DH systems models (see, e.g., \cite{DePersis2011}, \cite{Hauschild2020}). 
%
%\item In the recent review \cite{vandermeulen_control_review_18} a number of motivating arguments are presented for the development of advance control techniques for improving the performance of modern DH systems, for example by enabling fair heat distribution in the face of heat deficit or reducing the distribution losses by lowering the system's temperatures.
%
%
%\end{itemize}
%}

\subsection{Notation (including table of variables)}
{\color{black}
The symbol $\mathbb{R}$ denotes the set of real numbers. For a vector $x\in\mathbb{R}^n$,  $x_i$ denotes its $i$th component, {i.e.}, $x=[x_1,\dots, x_n]^\top$; moreover,  $\mathbf{sign}(x)=[\text{sign}(x_1),\dots,\text{sign}(x_n)]^\top $, with $\text{sign}(0)=0$,  and $\vert x \vert = [\vert x_1 \vert,\dots, \vert x_n \vert]^\top$. An $m\times n$ matrix with all-zero entries is written as $\boldsymbol{0}_{m\times n}$. An $n$-vector of ones is written as $\boldsymbol{1}_n$, whereas the identity matrix of size $n$ is represented by $I_n$. For any vector $x\in \mathbb{R}^n$, we denote by {$\mathrm{diag} (x) $}  a diagonal matrix with elements $x_i$ in its main diagonal. For any time-varying signal $w$, we represent by $\bar w$ its steady-state value, if exists. {Also, we write time derivatives as $\dot{x}(t)$, and omit the argument $t$ whenever is clear from  the context. For conveience, we provide summarize lists of relevant abbreviations and symbols in Tables \ref{table:abbreviations} and \ref{table:symbols}, respectively.}

{\footnotesize
 \begin{table}[t]
	\caption{List of Abbreviations}
	\centering
	{\color{black}
\begin{tabular}{ll}
DHS & district heating system\\
DN & distribution network\\
ST & storage tank\\
HX & heat exchanger\\
p & identifier for producers\\
c & identifier for consumers\\
st & identifier for storage tanks\\
sh & identifier for hot layer of storage tanks\\
sc & identifier for cold layer of storage tanks\\
s & identifier for elements associated to the DN's supply layer\\
r & identifier for elements associated to the DN's return layer
\end{tabular}}
\label{table:abbreviations}
\end{table}
}

{\footnotesize
 \begin{table}[t]
	\caption{Network Parameters}
	\centering
	{\color{black}
\begin{tabular}{ll}
$\mathcal{G}_\mathrm{s}$, $\mathcal{N}_\mathrm{s}$, $\mathcal{E}_\mathrm{s}$ & DH system's graph, nodes (junctions) and edges (pipes)\\
& associated to the supply layer\\
$\mathcal{G}_\mathrm{r}$, $\mathcal{N}_\mathrm{r}$, $\mathcal{E}_\mathrm{r}$ & DH system's graph, nodes (junctions) and edges (pipes)\\
& associated to the return layer\\
$P_{\mathrm{p},i}$ & power injection by $i$th producer, W\\
$P_{\mathrm{c},i}$ & power extraction by $i$th consumer, W\\
$q_{\mathrm{p},i}$ & flow rate  through edge the $i$th producer, $\mathrm{m}^3/s$\\
$q_{\mathrm{st},i}$ & flow rate at hot layer outlet (cold layer inlet) of the\\
& $i$th storage tank,  $\mathrm{m}^3/s$\\
$q_{\mathrm{c},i}$ & flow rate  through edge the $i$th consumer, $\mathrm{m}^3/s$\\
$q_{\mathrm{s},i}$ & flow rate  through $i\in \mathcal{E}_\mathrm{s}$, $\mathrm{m}^3/s$\\
$q_{\mathrm{r},i}$ & flow rate  through $i\in \mathcal{E}_\mathrm{r}$, $\mathrm{m}^3/s$\\
$V_{\mathrm{p},i}$ & effective volume of the secondary side of the\\
& $i$th producer's HX, $\mathrm{m}^3$\\
$V_{\mathrm{c},i}$ & effective volume of the primary side of the\\
& $i$th consumer's HX, $\mathrm{m}^3$\\
$V_{\mathrm{sh},i}$ & volume of water in the $i$th ST's hot layer, $\mathrm{m}^3$\\
$V_{\mathrm{c},i}$ &  volume of water in the $i$th ST's cold layer, $\mathrm{m}^3$\\
$V_{\mathrm{s},i}$ & effective volume of $i\in \mathcal{G}_\mathrm{s}$, $\mathrm{m}^3$\\
$V_{\mathrm{r},i}$ & effective volume of $i\in \mathcal{G}_\mathrm{r}$, $\mathrm{m}^3$\\
$T_{\mathrm{p},i}$ & average temperature of the secondary side of the\\
& $i$th producer's HX, $~^\circ \mathrm{C}$\\
$T_{\mathrm{c},i}$ & average temperature of the primary side of the\\
& $i$th consumer's HX, $~^\circ \mathrm{C}$\\
$T_{\mathrm{sh},i}$ & average temperature of the $i$th ST's hot layer, $~^\circ \mathrm{C}$\\
$T_{\mathrm{sc},i}$ & average temperature of the $i$th ST's cold layer, $~^\circ \mathrm{C}$\\
$T_{\mathrm{s},i}$ & average temperature of the $i\in \mathcal{G}_\mathrm{s}$, $~^\circ \mathrm{C}$\\
$T_{\mathrm{r},i}$ & average temperature of the $i\in \mathcal{G}_\mathrm{r}$, $~^\circ \mathrm{C}$\\
$\rho$ & density of water, $\mathrm{kg}/\mathrm{m}^3$\\
$c_{\mathrm{s.h.}}$ & specific heat of water, $\mathrm{J}/\mathrm{kg}~^\circ \mathrm{C}$\\
$(\cdot)^\mathrm{in}_i$ & quantity associated to an inlet\\
$(\cdot)^\mathrm{out}_i$ & quantity associated to an outlet
\end{tabular}}
\label{table:symbols}
\end{table}
}}

\section{System model}\label{sec:model}

\subsection{Setup and main modeling assumptions}

We consider a DH system with multiple, distributed $n_\mathrm{p}$ producers and $n_\mathrm{c}$ consumers  interconnected through a distribution network (DN) that has a supply (hot) layer and return (cold) layer.  The specific composition of producers and consumers is shown in  Fig.~\ref{fig:hydraulic_complex_elements}. Note that each producer can continuously drain water from the DN's return layer, heats it through a heat exchanger {\color{black}(HX)} and injects the heated stream into the DN's supply layer. A converse operation follows for consumers.
%
%Following \cite{Scholten_tcst_2015} and \cite{bendtsen_control_2017}, for producers we consider only the secondary side of their heat exchangers, whereas for consumers we only account the primary sides; details behind the plausibility of this assumption can be found in \cite[Appendix~A]{Scholten_tcst_2015}.  

Following \cite{Scholten_tcst_2015}, we consider a DH system with stratified storage tanks.  Each tank  stores a mixture of hot and cold water perfectly separated by a  thermocline. The volume of  hot water   is on top and the cold one at the bottom. It is assumed that there is no heat {or mass} exchange between the mixtures. Moreover, each storage tank is considered to have two inlet/outlet pairs for hot and cold water, respectively. The topology of storage tanks is shown in Fig.~\ref{fig:hydraulic_complex_elements}.  As a simplifying assumption, we consider that each producer is interfaced to the DN  via a storage tank. Then,   each producer  drains water from the cold layer of a storage tank and injects it into the tank's hot layer. Using  the remaining inlet/outlet pair, the tank  can fill in its cold layer with water taken from the return layer of the DN. At the same time, the storage tank injects water from its hot layer into the DN's supply layer. We note that our results can be slightly adjusted to consider producers directly connected to the DN. However, standalone storage tanks, {i.e.}, storage tanks with no immediate access to a heat producer are not considered in this work. 

{\color{black}
Additional modeling assumptions, some of which are fairly standard in related literature (see, e.g., \cite{Scholten_tcst_2015,bendtsen_control_2017,Alisic2019}) are the following (see \cite{jmc_dh_modeling_2020} for more details):
\begin{assumption}\label{assu:modeling_assumptions}
(i) the density $\rho>0$ and specific heat $c_{\mathrm{s.h.}}>0$ of water are spatially uniform and constant in time (for ease of notation we take $\rho=c_{\mathrm{s.h.}}=1$); (ii) the flow through any pipe is (spatially) one-dimensional. (iii) each device (pipe, storage tank, junction) is completely filled with water at all times; (iv) the internal energy of any water stream portion depends linearly on its temperature; {\color{black} the overall DH system is leak-free and lossless.}
\end{assumption}
}

%\begin{figure}
%\begin{center}
%\includegraphics[width=0.8\linewidth]{sketch_v1}
%\caption{\footnotesize Sketch based on  \cite{wang_meshed_17} of a DH system with  3 heat producers and 9 consumers.}
%\label{fig:dh_meshed}
%\end{center}
%\end{figure}

\begin{figure}[t]
\begin{center}
\includegraphics[width=0.98\linewidth]{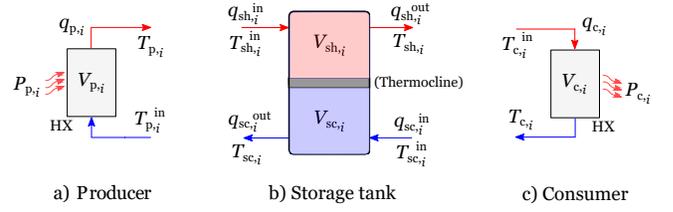}
\caption{\footnotesize Topologies of producers, consumers and storage tanks~\cite{Scholten_tcst_2015,bendtsen_control_2017}:   $T_{\chi,i}$, $V_{\chi,i}$,  $q_{\chi,i}$ and $P_{\chi,i}$ stand for temperature, volume, flow rate and thermal power associated with the $i$th device $\chi\in \{\mathrm{p},~\mathrm{sh},~\mathrm{sc},~\mathrm{st},~\mathrm{c} \}$.}
\label{fig:hydraulic_complex_elements}
\end{center}
\end{figure}

\subsection{Dynamics of  producers and storage tanks}

Consider the notation used in Fig.~\ref{fig:hydraulic_complex_elements}. {\color{black}
The heat balance at the secondary side of the $i$th producer's  heat exchanger can be written as follows \cite{Scholten_tcst_2015}, \cite{bendtsen_control_2017}:
\begin{equation}
V_{\mathrm{p},i}\dot{T}_{\mathrm{p},i} = q_{\mathrm{p},i}\left(T_{\mathrm{p},i}^\mathrm{in}-T_{\mathrm{p},i} \right) + P_{\mathrm{p},i},
\end{equation} 
where $q_{\mathrm{p},i}$ and $P_{\mathrm{p},i}$ are considered to be control variables: each $q_{\mathrm{p},i}$ could be controlled in practice  independently by a  pump in series with the secondary side of the producer's heat exchanger whereas each $P_{\mathrm{p},i}$ could be controlled independently by adjusting the flow in the primary side (hotter side)  of a producer's heat exchanger (see \cite[Appendix~A]{Scholten_tcst_2015} for additional details).

On the other hand, the heat balance at the hot and cold layer of each storage tank can be written respectively as follows:
\begin{align}
\frac{\mathrm{d}}{\mathrm{d}t}(V_{\mathrm{sh},i}{T}_{\mathrm{sh},i}) & = q_{\mathrm{sh},i}^\mathrm{in}T_{\mathrm{sh},i}^\mathrm{in}-q_{\mathrm{sh},i}^\mathrm{out}T_{\mathrm{sh},i}\\
\frac{\mathrm{d}}{\mathrm{d}t}(V_{\mathrm{sc},i}{T}_{\mathrm{sc},i}) & =q_{\mathrm{sc},i}^\mathrm{in}T_{\mathrm{sc},i}^\mathrm{in}-q_{\mathrm{sc},i}^\mathrm{out}T_{\mathrm{sc},i}.
\end{align}
Moreover, the volume dynamics of of the layers of each tank obey the following equations:
\begin{align}
\dot{V}_{\mathrm{sh},i} & = q_{\mathrm{sh},i}^\mathrm{in}-q_{\mathrm{sh},i}^\mathrm{out}\\
\dot{V}_{\mathrm{sc},i} & = q_{\mathrm{sc},i}^\mathrm{in}-q_{\mathrm{sc},i}^\mathrm{out}.
\end{align}

Without loss of generality, let the $i$th producer be adjacent to the $i$th storage tank. Since each producer is interfaced to the DN through a storage tank, we have that   $q_{\mathrm{sh},i}^\mathrm{in}=q_{\mathrm{sc},i}^\mathrm{out}=q_{\mathrm{p},i}$, $T_{\mathrm{pr},i}^\mathrm{in}=T_{\mathrm{sc},i}$ and  $T_{\mathrm{sh},i}^\mathrm{in}=T_{\mathrm{p},i}$. Also, since the storage tank must remain completely filled with water all the time, we have that $\dot{V}_{\mathrm{sh},i}+\dot{V}_{\mathrm{sc},i}=0$ must hold. This is equivalent to
\begin{equation}
q_{\mathrm{sc},i}^\mathrm{in}=q_{\mathrm{sh},i}^\mathrm{out}=:q_{\mathrm{st},i}.
\end{equation}
}
In view of the above considerations,  the temperature dynamics of the producers and storage tanks, as well as the volume dynamics of the storage tanks can be modeled as follows \cite{Scholten_tcst_2015} (see also \cite{bendtsen_control_2017,jmc_dh_modeling_2020}):
\begin{subequations}\label{eq:producer+storage dyn}
\begin{align}
V_{\mathrm{p},i}\dot{T}_{\mathrm{p},i} &  = q_{\mathrm{p},i}\left(T_{\mathrm{sc},i}-T_{\mathrm{p},i} \right) + P_{\mathrm{p},i},  \label{eq:producer+storage dyn (a)}\\
V_{\mathrm{sh},i}\dot{T}_{\mathrm{sh},i} & = q_{\mathrm{p},i}\left(T_{\mathrm{p},i}-T_{\mathrm{sh},i}\right), \label{eq:producer+storage dyn (b)}\\
V_{\mathrm{sc},i}\dot{T}_{\mathrm{sc},i} & =q_{\mathrm{st},i}\left(T_{\mathrm{sc},i}^\mathrm{in}-T_{\mathrm{sc},i}\right), \label{eq:cold storage dyn}\\
\dot{V}_{\mathrm{sh},i} & = q_{\mathrm{p},i}-q_{\mathrm{st},i} \label{eq:producer+storage dyn (c)},\\
\dot{V}_{\mathrm{sc},i} & = q_{\mathrm{st},i}-q_{\mathrm{p},i}. \label{eq:producer+storage dyn (d)}
\end{align}
\end{subequations}
We assume that $q_{\mathrm{p},i}$ and $P_{\mathrm{p},i}$ are control variables {\color{black}\cite{Scholten_tcst_2015,bendtsen_control_2017}} and $T_{\mathrm{sc},i}^\mathrm{in}$ is an external input. Later, when we introduce the temperature dynamics of the return layer of the DN, $T_{\mathrm{sc},i}^\mathrm{in}$ will be related with the temperature of a given junction in the return layer of the DN. We also identify each $q_{\mathrm{st},i}$ as an independent control variable, with the exception of one (see Remark~\ref{rem:independent_flows} for details about this).

\subsection{Consumer temperature dynamics}

Analogously to producers, the  heat balance at each consumer's heat exchanger is given by \cite{Scholten_tcst_2015}, \cite{bendtsen_control_2017}:
\begin{equation}\label{eq:consumer dynamics}
V_{\mathrm{c},i}\dot{T}_{\mathrm{c},i}  = q_{\mathrm{c},i}\left( T_{\mathrm{c},i}^\mathrm{in}-T_{\mathrm{c},i} \right) - P_{\mathrm{c},i}.
\end{equation}
The flow rate $q_{\mathrm{c},i}$ is an independent control variable, while $T_{\mathrm{c},i}^\mathrm{in}$ and $P_{\mathrm{c},i}\geq 0$ are external inputs. The power load $P_{\mathrm{c},i}$ will be treated as an {\em unknown} constant disturbance. Later we will equate $T_{\mathrm{c},i}^\mathrm{in}$ with the temperature of a certain junction in the supply layer of the DN.

\subsection{Distribution network's temperature dynamics}

Following  \cite{DePersis2011,wang_meshed_17,Hauschild2020}, we represent the supply and return layers of the DN as  connected  graphs with no self-loops. For the supply layer we introduce $\mathcal{G}_{\mathrm{s}}=(\mathcal{N}_\mathrm{s},\mathcal{E}_\mathrm{s})$, where the set of edges $\mathcal{E}_\mathrm{s}$ represents all distribution pipes, and the set of nodes $\mathcal{N}_\mathrm{s}$ denotes pipe junctions. An analogous description follows for the return layer of the DN, for which we  use the notation $\mathcal{G}_{\mathrm{r}}=(\mathcal{N}_\mathrm{r},\mathcal{E}_\mathrm{r})$. The focus of  this work is on DNs in which the supply and return layers are symmetric. Then, we assume that $\mathcal{G}_\mathrm{s}$ and $\mathcal{G}_\mathrm{r}$ are isomorphic and the bijection between $\mathcal{N}_\mathrm{s}$ and $\mathcal{N}_\mathrm{r}$ is referred to as $\gamma_\mathrm{dn}$. We refer the reader to \cite{felix_4gdh_21} for a discussion of prospective DH systems with non-symmetric DNs.\footnote{W.l.o.g., we assume that  any two pipes $(i,j)\in \mathcal{E}_\mathrm{s}$ and $\left(\gamma_\mathrm{dn}(i),\gamma_\mathrm{dn}(j)\right)\in \mathcal{E}_\mathrm{r}$ have the same length and diameter.}

{\color{black}We identify with the $i$th storage tank a unique node $k\in \mathcal{N}_\mathrm{s}$ such that there is a stream with rate $q_{\mathrm{st},i}$ from  the tank's hot layer into $k$, and at the same time there is a stream (with the same rate) from $\gamma_\mathrm{dn}(k)\in \mathcal{N}_\mathrm{r}$ towards  the tank's cold layer.  Analogously, to the $i$th consumer we associate a unique node $k\mathcal{N}_\mathrm{s}$  such that the there is a stream with rate $q_{\mathrm{c},i}$ from $k$  towards the consumer's heat exchanger and finally reaching the node $\gamma_{\mathrm{dn}}(k)\in \mathcal{N}_\mathrm{r}$.}
% {\color{red}Out of simplicity, we assume that  there is  most one storage tank injecting (draining) water at any $k\in \mathcal{N}_\mathrm{s}$ ($k\in \mathcal{N}_\mathrm{r}$). Analogously, we consider that there is at most one consumer injecting (draining) water into any $k\in \mathcal{N}_\mathrm{r}$ ($\mathcal{N}_\mathrm{s}$).} {\color{black}I think this is not necessary.}

For ease of presentation, let us introduce further notation using $\mathcal{G}_\mathrm{s}$ as reference. We  fix an arbitrary reference orientation to every edge of $\mathcal{G}_\mathrm{s}$. Then, for any $i\in\mathcal{E}_\mathrm{s}$ with end nodes $j,k\in \mathcal{N}_\mathrm{s}$, $j\neq k$, we  say that $j$  is the head and $k$ is the tail of $i$, or viceversa, that $j$  is the tail and $k$ is the head of $i$. Moreover, following  \cite{Hauschild2020,krug_DH_2021,valdimarsson_14}, we introduce the functions $\mathcal{N}_\mathrm{s}^{-},\mathcal{N}_\mathrm{s}^{+}:\mathcal{E}_\mathrm{s}\rightarrow \mathcal{N}_\mathrm{s}$ and $\mathcal{E}_\mathrm{s}^{-},\mathcal{E}_\mathrm{s}^{+}:\mathcal{N}_\mathrm{s}\rightarrow \mathcal{E}_\mathrm{s}$ defined as follows. For any $i\in \mathcal{E}_\mathrm{s}$, $\mathcal{N}_\mathrm{s}^{-}(i)$ and  $\mathcal{N}_\mathrm{s}^{+}(i)$ respectively denote the tail and head of $i$; also, for any $j\in \mathcal{N}_\mathrm{s}$, $\mathcal{E}_{\mathrm{s}}^{-}(j)$ and $\mathcal{E}_\mathrm{s}^{+}(j)$ denote sets of edges with $j$ as tail node and $j$ as head node, respectively. To streamline the subsequent definition of the DN's temperature dynamics, we  assume that the reference orientation of any edge $i\in \mathcal{E}_\mathrm{s}$ matches the direction of the stream through it.  That is, if $j,k\in \mathcal{N}_\mathrm{s}$, $j\neq k$, are the tail and head of any $i\in \mathcal{E}_\mathrm{s}$, respectively,  then  the stream through $i$, henceforth denoted by $q_{\mathrm{s},i}$, is assumed to flow from $j$ to $k$ and we consider that $q_{\mathrm{s},i}\geq 0$.\footnote{Since flow reversals may occur, or the flow through an edge may simply not match the edge's reference orientation,  $q_\mathrm{s}$-dependent functions analogous to $\mathcal{N}_\mathrm{s}^{-}$, $\mathcal{N}_\mathrm{s}^{+}$, $\mathcal{E}_\mathrm{s}^{-}$ and $\mathcal{E}_\mathrm{s}^{+}$ can be defined to identify the source and target node of the stream through any edge (see, e.g., \cite{Hauschild2020}.)}
%\todo[inline]{I dont think $q_{s,i}$ is defined in text}

{\color{black}
Considering the above definitions and assumption, now we can write the temperature dynamics of the supply layer of the DN. First, the heat balance at each $i\in \mathcal{E}_\mathrm{s}$ can be written as follows (see \cite{jmc_dh_modeling_2020} for more details):
\begin{align}\label{eq:model_pipe_no_heat_trans}
V_{\mathrm{s},i}  \dot{T}_{\mathrm{s},i}   =   q_{\mathrm{s},i} \left(T_{\mathrm{s},i}^\mathrm{in}-T_{\mathrm{s},i}^\mathrm{out} \right).
\end{align}
Note that the right-hand side of this equation is the heat contribution due to the stream entering the pipe at temperature $T_{\mathrm{s},i}^\mathrm{in}$ minus the heat loss due to the stream leaving the pipe at temperature $T_{\mathrm{s},i}^\mathrm{out}$. Following  \cite{Hauschild2020} and \cite{krug_DH_2021} we impose the following constraints for all $i\in \mathcal{E}_\mathrm{s}$
\begin{align}\label{eq:node_constraints_temps}
T_{\mathrm{s},i}^\mathrm{in}   = T_{\mathrm{s},j}\left. \right\vert_{j=\mathcal{N}_\mathrm{s}^{-}(i)},~~~ {T_{\mathrm{s},i}^\mathrm{out}}   = {T_{\mathrm{s},i}},
\end{align}
which transform \eqref{eq:model_pipe_no_heat_trans} into
\begin{equation}\label{eq:pipe_simplified_wcons}
V_{\mathrm{s},i}  \dot{T}_{\mathrm{s},i}   =   q_{\mathrm{s},i} \left(T_{\mathrm{s},j}-T_{\mathrm{s},i} \right)\left. \right\vert_{j=\mathcal{N}_\mathrm{s}^{-}(i)},~ \forall i\in \mathcal{E}_\mathrm{s}.
\end{equation}
We note that the constraints \eqref{eq:node_constraints_temps} respectively imply that the stream entering any pipe will have the temperature of the node from which the stream sources from and that the temperature of the stream at the outlet of any pipe will have the same temperature as the spatially-averaged  temperature of the pipe's control volume (upwind scheme).

In view of  \eqref{eq:node_constraints_temps},  the heat balance at each  $j\in \mathcal{N}_\mathrm{s}$ is equivalent to the following
\begin{align}\label{eq:thermal_model_detail_nodes_0}
V_{\mathrm{s},j}\dot{T}_{\mathrm{s},j} & = \sum_{i\in \mathcal{E}^{+}_\mathrm{s}(j)}  q_{\mathrm{s},i}T_{\mathrm{s},i} 
- \left(\sum_{i\in \mathcal{E}_\mathrm{s}^{-}(j)}q_{\mathrm{s},i}\right)T_{\mathrm{s},j} \nonumber \\
& ~~~ +\sum_{i=1}^{n_\mathrm{p}} \alpha_{i,j}q_{\mathrm{st},i}T_{\mathrm{sh},i} - \left(\sum_{i=1}^{n_\mathrm{c}}\beta_{i,j} q_{\mathrm{c},i}\right) T_{\mathrm{s},j},
\end{align}
where $\alpha_{i,j}=1$ if $j\in \mathcal{N}_\mathrm{s}$ receives a stream from the $i$th storage tank (and $\alpha_{i,j}=0$ otherwise). Analogously,   $\beta_{i,j}=1$ if from $k\in \mathcal{N}_\mathrm{s}$ a stream is directed towards the  $i$th consumer. We note that the term in the left-hand side of \eqref{eq:thermal_model_detail_nodes_0} represents the rate of change of the thermal energy stored at node $k$ whereas  the terms in the right-hand side are the sum of the thermal energies of the streams that target $k$  or source $k$.   Note that the energy exchange due to the interaction with storage tanks and consumers is accounted separately (see the presence of the flows $q_{\mathrm{st},i}$ and $q_{\mathrm{c},i}$). 

One further constraint is due to volume (mass) balance at each $j\in\mathcal{N}_\mathrm{s}$, which reads as follows:
\begin{align}\label{eq:node_constraints_c}
 \dot{V}_{\mathrm{s},j} = 0  & = \sum_{i\in \mathcal{E}_\mathrm{s}^{+}(j)}  q_{\mathrm{s},i} - \sum_{i \in \mathcal{E}_\mathrm{s}^{-}(j)} q_{\mathrm{s},i}  +\sum_{i=1}^{n_\mathrm{p}} \alpha_{i,j}q_{\mathrm{st},i} \nonumber \\
& ~~~  - \sum_{i=1}^{n_\mathrm{c}}\beta_{i,j} q_{\mathrm{c},i}.
\end{align}
Clearing $ \sum_{i \in \mathcal{E}_\mathrm{s}^{-}(j)} q_{\mathrm{s},i}$ from \eqref{eq:node_constraints_c} and substituting into \eqref{eq:thermal_model_detail_nodes_0} results in the following simplification:
\begin{align}\label{eq:thermal_model_detail_nodes}
 V_{\mathrm{s},j} \dot{T}_{\mathrm{s},j}  & = \sum_{i\in \mathcal{E}_\mathrm{s}^{+}(j)}  q_{\mathrm{s},i}\left(T_{\mathrm{s},i} -T_{\mathrm{s},j}\right), \nonumber \\
 & ~~~~+\sum_{i=1}^{n_\mathrm{p}} \alpha_{i,j}q_{\mathrm{st},i}\left(T_{\mathrm{sh},i}-T_{\mathrm{s},j}\right),~\forall j\in \mathcal{N}_\mathrm{s}.
\end{align}
Then, equations \eqref{eq:pipe_simplified_wcons} and \eqref{eq:thermal_model_detail_nodes} conform the model for the temperature dynamics of the supply layer of  DH system's DN.

For the return layer of the DN, analogous definitions, assumptions and computations can be introduced to obtain
\begin{subequations}\label{eq:return layer temp dyn}
\begin{align}
V_{\mathrm{r},i}  \dot{T}_{\mathrm{r},i}  &  =   q_{\mathrm{r},i} \left(T_{\mathrm{r},j}-T_{\mathrm{r},i} \right)\left. \right\vert_{j=\mathcal{N}_\mathrm{r}^{-}(i)},~ \forall i\in \mathcal{E}_\mathrm{r}\\
 V_{\mathrm{r},j} \dot{T}_{\mathrm{r},j}  & = \sum_{i\in \mathcal{E}_\mathrm{r}^{+}(j)}  q_{\mathrm{r},i}\left(T_{\mathrm{r},i} -T_{\mathrm{r},j}\right), \nonumber \\
 & ~~~~+\sum_{i=1}^{n_\mathrm{c}} \beta_{i,j}q_{\mathrm{c},i}\left(T_{\mathrm{c},i}-T_{\mathrm{r},j}\right),~\forall j\in \mathcal{N}_\mathrm{r},
\end{align}
\end{subequations}
where $q_{\mathrm{r},i}$ is the flow rate of the stream through any edge $i\in \mathcal{E}_\mathrm{r}$.
}

{\color{black} For ease of reference, we find it convenient to write compactly the DN's dynamics \eqref{eq:pipe_simplified_wcons}, \eqref{eq:thermal_model_detail_nodes}  and \eqref{eq:return layer temp dyn} as follows.}  The temperature dynamics of each $i\in \mathcal{E}_{\chi}$, $j\in\mathcal{N}_{\chi}$, $\chi\in\{ \mathrm{s},~\mathrm{r}\}$ can be compactly represented as:
 \begin{subequations}\label{eq:supply_return DN temp dyn}
\begin{align}
 V_{{\chi},i}  \dot{T}_{{\chi},i}    & =   q_{{\chi},i} \left(T_{{\chi},j}-T_{{\chi},i} \right)\left. \right\vert_{j=\mathcal{N}_{\chi}^{-}(i)}, \label{eq:pipe_dyn_DN}\\
  V_{{\chi},j} \dot{T}_{{\chi},j}  & = \sum_{k\in \mathcal{E}_{\chi}^{+}(j)}  q_{{\chi},k}\left(T_{{\chi},k} -T_{{\chi},j}\right) + \Phi_{\chi,j}, \label{eq:node_dyn_DN}\\
 \Phi_{\chi,j} & = \begin{cases} 
 \sum_{k=1}^{n_\mathrm{p}} \alpha_{k,j}q_{\mathrm{st},k}\left(T_{\mathrm{sh},k}-T_{\mathrm{s},j}\right), & \chi=\mathrm{s},\\
 \sum_{k=1}^{n_\mathrm{c}} \beta_{k,j}q_{\mathrm{c},k}\left(T_{\mathrm{c},k}-T_{\mathrm{r},j}\right), & \chi=\mathrm{r},
 \end{cases}\label{eq:node_dyn_extra}
\end{align}
\end{subequations}
where $V_{\chi,i}$, $T_{\chi,i}$ and $q_{\chi,i}$ respectively stand for volume, temperature  and flow rate of the respective elements in $\mathcal{G}_\chi$. {\color{black} Also, we recall that} $\alpha_{k,j}=1$ if $j\in \mathcal{N}_\mathrm{s}$ receives a stream from the $k$th storage tank (and $\alpha_{k,j}=0$ otherwise). Analogously,   $\beta_{k,j}=1$ if from $j\in \mathcal{N}_\mathrm{s}$ a stream is directed towards the  $k$th consumer.

Equation~\eqref{eq:pipe_dyn_DN} represents the heat balance at any pipe  $i\in \mathcal{E}_{\chi}$, in which we have used  the boundary conditions $T_{\mathrm{s},i}^\mathrm{in}   = T_{\mathrm{s},j}\left. \right\vert_{j=\mathcal{N}_\mathrm{s}^{-}(i)}$ and  ${T_{\mathrm{s},i}^\mathrm{out}}   = {T_{\mathrm{s},i}}$  \cite{Hauschild2020,krug_DH_2021},  meaning that the stream entering any pipe will have the temperature of the node from which the stream sources from and that the temperature of the stream at the outlet of any pipe will have the same temperature as the spatially-averaged  temperature of the pipe's control volume (upwind scheme). Equation~\eqref{eq:node_dyn_DN} models the heat balance at each node $j\in \mathcal{N}_\chi$. The term in the left-hand side of \eqref{eq:node_dyn_DN} is the rate of change of the thermal energy stored at  $j$ and in the right-hand side we have  the sum of the thermal energies of the streams that target or source from $j$.   The interaction with storage tanks and consumers is represented by the term $\Phi_{\chi,j}$. Further details  appear in~\cite{jmc_dh_modeling_2020}.

\begin{remark}\label{rem:independent_flows}
{\bf (I)} Following \cite{jmc_dh_modeling_2020} (see also \cite{wang_meshed_17}),  we take as independent variables each $q_{\mathrm{p},i}$,  $q_{\mathrm{c},j}$ and $q_{\mathrm{s},k}$ ($q_{\mathrm{r},k}$) associated with any edge of $\mathcal{G}_\mathrm{s}$ ($\mathcal{G}_\mathrm{r}$) being a {\em chord}.  The same can be done for each $q_{\mathrm{st},i}$, except for one, say for the $m$th tank. Such a constraint stems  from the need to meet Kirchhoff's current laws and has the implication that $q_{\mathrm{st},m}=\sum_{\forall i} q_{\mathrm{c},i}-\sum_{\forall j\neq m}q_{\mathrm{st},j}$.  {\bf (II)} Having defined the temperature dynamics of the DN, we can define $T_{\mathrm{sc},i}^\mathrm{in}$ and $T_{\mathrm{c},i}^\mathrm{in}$  in \eqref{eq:producer+storage dyn} and \eqref{eq:consumer dynamics}, respectively,  as follows:
\begin{equation}\label{eq:temp inlet storage and consumers}
T_{\mathrm{sc},i}^\mathrm{in}=\alpha_{i,j}T_{\mathrm{r},k}\vert_{k=\gamma_{\mathrm{dn}}(j)},~T_{\mathrm{c},i}^\mathrm{in}=\beta_{i,j}T_{\mathrm{s},k}\vert_{k=\gamma_{\mathrm{dn}}^{-1}(j)}.
\end{equation}
Therefore, the overall temperature dynamics of the DH system are given by \eqref{eq:producer+storage dyn}, \eqref{eq:consumer dynamics}, \eqref{eq:supply_return DN temp dyn} and \eqref{eq:temp inlet storage and consumers}.
\end{remark}

{{\color{black}}Considering our assumption that storage tanks operate at maximum capacity all the time (constant total volume) and that the overall district heating system has a constant volume and is leak-free (see new Assumption~\ref{assu:modeling_assumptions}), then conservation of mass dictates that the sum of the flows entering/leaving  the district heating's supply/return layer should be equal to the sum of flows leaving/entering it, or equivalently, that
\begin{equation*}
\sum_{i=1}^{n_\mathrm{pr}}{q_{\mathrm{st},i}}+\sum_{i=1}^{n_\mathrm{c}}q_{\mathrm{c},i}=0,
\end{equation*}
where an adequate convention for the sign and direction of the flows should be taken. It is explained in \cite{jmc_dh_modeling_2020} that all consumer flows can be chosen as independent variables. Then, the equation above explains how one flow $q_{\mathrm{st},m}$ is a dependent variable (see Remark~\ref{rem:independent_flows}.{\bf{(I)}}). The remaining flows $q_{\mathrm{st},j}$, $j\neq m$, can also be chosen as independent variables~\cite{jmc_dh_modeling_2020}.}

%\subsection{Overall dynamics in vector form (if needed)}
%
%Some info needed:
%
%
%We define a constant incidence matrix $\mathcal{B}_0$ associated with the arbitrary orientation we have fixed for the DH system's edges, as follows:
%\begin{equation}\label{eq:incidence_matrix_B0}
%{(\mathcal{B}_0)_{i,j}=}\resizebox{0.3\textwidth}{!}{${\begin{cases}
%1{,} & \text{if node $i$  is the head of edge  $j$},\\
%-1{,} & \text{if node $i$  is the tail of edge  $j$},\\
%0{,} & \text{otherwise}.
%\end{cases}}$}
%\end{equation}

%\section{Control objectives}\label{sec:control_objective} 
%
%{\color{red}
%\begin{enumerate}
%
%\item Enumerate the objectives (maybe in the form of a list).
%
%\item  Comment of the relevancy or physical meaningfulness  of the objectives.
%
%\item In a nutshell, describe our approach to achieve these objectives.
%
%\end{enumerate}
%}

\section{Control design and stability analysis}\label{sec:control}

Control design is conducted in this section to meet the  objectives of regulating each producer supply temperature $T_{\mathrm{p},i}$, each consumer return temperature $T_{\mathrm{c},i}$ and each storage tank (hot layer) volume $V_{\mathrm{sh},i}$ towards constant setpoints, {\color{black}usually specified by the DH operator and possibly based on some optimization criteria}.

In this development, both the consumer and producer controllers are fully decentralised, requiring only local measurements for implementation. The advantage of a decentralised architecture over a distributed one is twofold. Firstly, the control design is independent of the network topology, and since only local measurements are required for control implementation, no communications are required among the controllers. Secondly, as stability is verified at the individual nodes, producers and consumers can be added or removed from the network without impacting the overall stability.

\subsection{Control of producer temperatures}\label{subsec:producer_temp}

First, we consider the temperature regulation of the producer. The objective
is to utilise the power inputs $P_{\mathrm{p},i}$ to regulate the temperature to a known
constant value $T_{\mathrm{p},i}^\star$.

\begin{proposition}\label{prop:producerTemp}
Consider the $i$th  producer's temperature dynamics \eqref{eq:producer+storage dyn (a)} in closed-loop with the control law
\begin{equation}\label{eq:controller P_pi}
P_{\mathrm{p},i}=-q_{\mathrm{p},i}\left(T_{\mathrm{sc},i}-T_{\mathrm{p},i} \right)-k_{\mathrm{p},i}\left(T_{\mathrm{p},i}-T_{\mathrm{p},i}^\star \right),
\end{equation}
where $k_{\mathrm{p},i}>0$ is a tuning parameter. The resulting closed-loop dynamics are given by
\begin{equation}\label{eq:closed-loop dyn producers temps}
V_{\mathrm{p},i}\dot{T}_{\mathrm{p},i}=-k_{\mathrm{p},i}\left(T_{\mathrm{p},i}-T_{\mathrm{p},i}^\star \right)
\end{equation}
and the producer temperature converges monotonically to $T_{\mathrm{p},i}^\star$ at an exponential rate.

\end{proposition}
\begin{proof}
The verification of the closed-loop dynamics \eqref{eq:closed-loop dyn producers temps} follows by direct substitution of \eqref{eq:controller P_pi} into \eqref{eq:producer+storage dyn (a)}. The temperature dynamics \eqref{eq:closed-loop dyn producers temps} have the solution $T_{\mathrm{p},i}(t)=T_{\mathrm{p},i}^\star+\left[T_{\mathrm{p},i}(0)-T_{\mathrm{p},i}^\star\right]e^{-k_{\mathrm{p},i}V_{\mathrm{p},i}^{-1}t}$, verifying the stability properties.
%
% To verify exponential stability, consider the Lyapunov candidate
%\begin{equation}
%W_{\mathrm{p},i}=\frac{1}{2}V_{\mathrm{p},i}\left(T_{\mathrm{p},i}-T_{\mathrm{p},i}^\star \right)^2.
%\end{equation}
%The time derivative of $W_{\mathrm{p},i}$ along solutions of \eqref{eq:closed-loop dyn producers temps} satisfies 
%\begin{align}
%\dot{W}_{\mathrm{p},i} & =\left(T_{\mathrm{p},i}-T_{\mathrm{p},i}^\star \right)V_{\mathrm{p},i}\dot{T}_{\mathrm{p},i}=-k_{\mathrm{p},i}\left(T_{\mathrm{p},i}-T_{\mathrm{p},i}^\star \right)^2.
%\end{align}
%By \cite[Theorem~4.10]{khalil_book_nls_02} the equilibrium $T_{\mathrm{p},i}=T_{\mathrm{p},i}^\star$ is globally exponentially stable.
\end{proof}

%\begin{remark}
%~
%\begin{enumerate}
%\item The exact solution of \eqref{eq:closed-loop dyn producers temps} is given by $T_{\mathrm{p},i}(t)=T_{\mathrm{p},i}^\star+e^{-k_{\mathrm{p},i}V_{\mathrm{p},i}^{-1}t}$.
%\smallskip
%
%\todo[inline]{Remove below if needing space}
%\item If the tank's cold layer temperature $T_{\mathrm{sc},i}$ or flow rate $q_{\mathrm{p},i}$ is not available  for control, it is possible to achieve input-to-state stability about the point $T_{\mathrm{p},i}^\star$ with the following modified controller:
%\begin{equation}
%P_{\mathrm{p},i}=@@@.
%\end{equation}
%Due to space constraints, we omit the proof of this claim.
%\end{enumerate}
%\end{remark}

\subsection{Storage tank control (hot layer)}
Next we consider regulating both the volume and temperature of each storage tank's 
hot layer. The objective is to regulate the volume $V_{\mathrm{sh},i}$ to a known constant value $V_{\mathrm{sh},i}^\star$ via control of the producer's flow rate $q_{\mathrm{p},i}$. As the temperature of the $i$th producer is regulated to the (specified) constant value $T_{\mathrm{p},i}^\star$, it is not surprising that the temperature of the tank's  hot layer converges to the same value.

{\color{black}The outgoing flow from each tank, $q_{\mathrm{st},i}$, could be controlled locally via a valve (or pump, see \cite{jmnj_2021_adaptive}). Each of these flows are treated as independent inputs with the exception for one node. As noted in Remark~\ref{rem:independent_flows}, there exists an index $m$ such that  $q_{\mathrm{st},m}=\sum_{\forall i} q_{\mathrm{c},i}-\sum_{\forall j\neq m}q_{\mathrm{st},j}$, which could potentially lead to negative $q_{\mathrm{st},m}$ in some scenarios. To avoid such a situation, we make the following assumption.}

\begin{assumption}\label{assu:positivity_flow_inputs}
The flows $q_{\mathrm{st},i}$ are non-negative at all times.
%\footnote{This property is partly ensured by subsequent control design. It is indicated in Remark~\ref{rem:independent_flows} that  most of these signals are independent control inputs, yet there exists an index $m$ such that  $q_{\mathrm{st},m}=\sum_{\forall i} q_{\mathrm{c},i}-\sum_{\forall j\neq m}q_{\mathrm{st},j}$, which could potentially lead to $q_{\mathrm{st},m}$ being negative. Nonetheless, in practice this can be prevented by placing a check valve at the hot water outlet of each tank.}
\end{assumption}

{\color{black}In subsequent design, each $q_{\mathrm{c},i}$ will be chosen to be non-negative, ensuring that $\sum_{\forall i} q_{\mathrm{c},i} \geq 0$. Note also that this assumption can be satisfied in practice by placing a check valve at the hot water outlet of each tank.}

\begin{proposition}\label{prop:vsh_controller}
Consider the volume dynamics of the $i$th tank's hot layer \eqref{eq:producer+storage dyn (c)} in closed-loop with the continuous control law
\begin{equation}\label{eq: controller q_p}
q_{\mathrm{p},i}=
\begin{cases}
-\kappa_{\mathrm{p},i}\left(V_{\mathrm{sh},i}-V_{\mathrm{sh},i}^\star \right) + q_{\mathrm{st},i}, & V_{\mathrm{sh},i}\leq V_{\mathrm{sh},i}^\star,\\
q_{\mathrm{st},i}e^{-\left(V_{\mathrm{sh},i}-V_{\mathrm{sh},i}^\star \right)}, & V_{\mathrm{sh},i}> V_{\mathrm{sh},i}^\star,
\end{cases}
\end{equation}
where $\kappa_{\mathrm{p},i}>0$ is a tuning parameter that adjusts the rate of convergence of $V_{\mathrm{sh},i}$ towards $V_{\mathrm{sh},i}^\star$. The resulting closed-loop dynamics are described by
\begin{equation}\label{eq:v_sh dyn cl}
\dot{V}_{\mathrm{sh},i}=
\begin{cases}
-\kappa_{\mathrm{p},i}\left(V_{\mathrm{sh},i}-V_{\mathrm{sh},i}^\star \right), & V_{\mathrm{sh},i}\leq V_{\mathrm{sh},i}^\star,\\
-\xi_{\mathrm{sh},i}(V_{\mathrm{sh},i})q_{\mathrm{st},i}, & V_{\mathrm{sh},i} > V_{\mathrm{sh},i}^\star,
\end{cases}
\end{equation} 
where $\xi_{\mathrm{sh},i}(V_{\mathrm{sh},i}) = \left(1 - e^{-\left(V_{\mathrm{sh},i}-V_{\mathrm{sh},i}^\star \right)}\right)$ is non-negative for $V_{\mathrm{sh},i} > V_{\mathrm{sh},i}^\star$.
The equilibrium point $\bar{V}_{\mathrm{sh},i}=V_{\mathrm{sh},i}^\star$ is Lyapunov stable and, in the case that ${q}_{\mathrm{st},i}>0$, asymptotically stable. 
\end{proposition}

\begin{proof}
The verification of the closed-loop dynamics \eqref{eq:v_sh dyn cl} follows by direct substitution of \eqref{eq: controller q_p} into   \eqref{eq:producer+storage dyn (c)}. To verify stability, consider the {\color{black}Lyapunov function candidate $
W_{V_{\mathrm{sh},i}}=\frac{1}{2}(V_{\mathrm{sh},i}-V_{\mathrm{sh},i}^\star)^2.
$}
Its time derivative along solutions of \eqref{eq:v_sh dyn cl} can be straightforwardly verified to satisfy $\dot{W}_{V_{\mathrm{sh},i}}=$
\begin{equation}
\begin{cases}
-\kappa_{\mathrm{p},i}\left(V_{\mathrm{sh},i}-V_{\mathrm{sh},i}^\star \right)^2, & V_{\mathrm{sh},i}\leq V_{\mathrm{sh},i}^\star,\\
-\xi_{\mathrm{sh},i}(V_{\mathrm{sh},i})q_{\mathrm{st},i}\left(V_{\mathrm{sh},i}-V_{\mathrm{sh},i}^\star  \right), & V_{\mathrm{sh},i} > V_{\mathrm{sh},i}^\star.
\end{cases}
\end{equation} 
By Assumption \ref{assu:positivity_flow_inputs}, $q_{\mathrm{st},i}$ is non-negative for all the time, implying that $W_{V_{\mathrm{sh},i}}$ is non-increasing along solutions of  \eqref{eq:v_sh dyn cl}, which verifies Lyapunov stability. Note that if $q_{\mathrm{st},i}(t)>0$, $W_{V_{\mathrm{sh},i}}$ is {\em strictly} decreasing, which implies asymptotic stability and convergence of $V_{\mathrm{sh},i}$ to $V_{\mathrm{sh},i}^\star$.
%\footnote{Since  \eqref{eq:producer+storage dyn (c)} and \eqref{eq:producer+storage dyn (d)} imply $\dot{V}_{\mathrm{sh},i}+\dot{V}_{\mathrm{sc},i}=0$ all the time, then boundedness and convergence of $V_{\mathrm{sc},i}$ are also guaranteed.}
\end{proof}

%\begin{remark}\label{rem:q_p for iss and stability v_sc}
%%
%%If the flow rate $q_{\mathrm{st},i}$ is unavailable for control purposes, input-to-state stability can be verified for  \eqref{eq:v_sh dyn cl} about the point $V_{\mathrm{sh}}=V_{\mathrm{sh},i}^\star$ with the control signal $q_{\mathrm{p},i}=-\kappa_{\mathrm{p},i}\left( V_{\mathrm{sh},i}-V_{\mathrm{sh},i}^\star\right)$.
%%
% Due to volume conservation within the overall DH system, (asymptotic) stability of  the  each storage tank's hot layer volume dynamics implies (asymptotic) stability of the cold layer's volume also; note from \eqref{eq:producer+storage dyn (c)} and \eqref{eq:producer+storage dyn (d)} that $\dot{V}_{\mathrm{sh},i}+\dot{V}_{\mathrm{sc},i}=0$ all the time. 
%\end{remark}

Next we show that the temperature of the $i$th storage tank's hot layer converges to the $i$th producer's (desired) outlet temperature $T_{\mathrm{p},i}^\star$. %It is assumed that $V_{\mathrm{sh},i}$ is strictly positive at the initial time which, by virtue of Proposition~\ref{prop:vsh_controller}, implies that this state remains positive all the time.

\begin{proposition}\label{prop:HotStorageTemperature}
Consider the temperature $T_{\mathrm{sh},i}$ of the $i$th storage tank and assume that the flow rate $q_{\mathrm{p},i}$ satisfies \eqref{eq: controller q_p}. Then, the temperature of the storage tank satisfies the following: {\bf (I)}  If the $i$th producer and storage tank have initial conditions satisfying $|T_{\mathrm{p},i}(0) - T_{\mathrm{p},i}^\star|, |T_{\mathrm{sh},i}(0) - T_{\mathrm{p},i}^\star| \leq \phi_i$, for some $\phi_i > 0$, then the storage tank temperature satisfies $|T_{\mathrm{sh},i}(t) - T_{\mathrm{p},i}^\star| \leq \phi_i$ all the time. {\bf (II)} If the $i$th  producer's flow rate $q_{\mathrm{p},i}$ is strictly positive, the temperature $T_{\mathrm{sh},i}$ converges to the $i$th producer's set-point $T_{\mathrm{p},i}^\star$.
% $|T_{sh,i}(t) - T_{p,i}^\star| \leq \phi_i$ all the time.
%\begin{enumerate}
%	\item\label{prop:HotStorageTemperature:1} If the $i$th producer and storage tank have initial conditions satisfying $|T_{p,i}(0) - T_{p,i}^\star|, |T_{sh,i}(0) - T_{p,i}^\star| \leq \phi_i$, for some $\phi_i > 0$, then the storage tank temperature satisfies
% $|T_{sh,i}(t) - T_{p,i}^\star| \leq \phi_i$ all the time.
%	
%	\item\label{prop:HotStorageTemperature:2} If the $i$th  producer's flow rate $q_{\mathrm{p},i}$ is strictly positive, the temperature $T_{\mathrm{sh},i}$ converges to the $i^{th}$ producer's set-point $T_{\mathrm{p},i}^\star$.
%\end{enumerate}
\end{proposition}

\begin{proof}
To verify stability of the temperature dynamics \eqref{eq:producer+storage dyn (b)}, consider the Lyapunov function
\begin{equation}
W_{T_{\mathrm{sh},i}}=\frac{1}{2}\left(T_{\mathrm{sh},i}-T_{\mathrm{p},i}^\star \right)^2.
\end{equation}
Its time derivative along solutions of  \eqref{eq:producer+storage dyn (b)}  satisfies:
\begin{align}\label{eq:time-der lyap TSH (a)}
 \dot{W}_{T_{\mathrm{sh},i}}  
%& =\left(T_{\mathrm{sh},i}- T_{\mathrm{p},i}^\star  \right)\dot{T}_{\mathrm{sh},i} \nonumber \\
% & = \left(T_{\mathrm{sh},i}- T_{\mathrm{p},i}^\star  \right)\frac{1}{V_{\mathrm{sh},i}}q_{\mathrm{p},i}\left(T_{\mathrm{p},i}-T_{\mathrm{sh},i} \right) \nonumber \\
% & =  -\left(T_{\mathrm{sh},i}- T_{\mathrm{p},i}^\star  \right)\frac{1}{V_{\mathrm{sh},i}}q_{\mathrm{p},i}\left(T_{\mathrm{sh},i}-T_{\mathrm{p},i}^\star + T_{\mathrm{p},i}^\star -T_{\mathrm{p},i}\right)\nonumber\\
 & = -\frac{1}{V_{\mathrm{sh},i}}q_{\mathrm{p},i}\left(T_{\mathrm{sh},i}-T_{\mathrm{p},i}^\star \right)^2 \nonumber \\
 & ~~~~ +\frac{1}{V_{\mathrm{sh},i}}q_{\mathrm{p},i}\left(T_{\mathrm{sh},i}-T_{\mathrm{p},i}^\star \right)\left(T_{\mathrm{p},i}-T_{\mathrm{p},i}^\star \right).
\end{align}
Applying  Young's inequality, $\dot{W}_{T_{\mathrm{sh},i}}$ satisfies:
\begin{align*}
& \dot{W}_{T_{\mathrm{sh},i}}  \leq  -\frac{q_{\mathrm{p},i}}{V_{\mathrm{sh},i}}\left(T_{\mathrm{sh},i}-T_{\mathrm{p},i}^\star \right)^2 \nonumber \\
& + \frac{\mu_iq_{\mathrm{p},i}}{2V_{\mathrm{sh},i}}\left(T_{\mathrm{sh},i}-T_{\mathrm{p},i}^\star \right)^2 + \frac{q_{\mathrm{p},i}}{2\mu_iV_{\mathrm{sh},i}}\left(T_{\mathrm{p},i}-T_{\mathrm{p},i}^\star \right)^2,
\end{align*} 
where $\mu_i>0$ is an arbitrary constant. Noting that $q_{\mathrm{p},i}$ is non-negative by definition \eqref{eq: controller q_p}, we have that $\tfrac{q_{\mathrm{p},i}}{V_{\mathrm{sh},i}}\geq 0$. Then, setting $\mu_i=1$ results in
\begin{align}\label{WTshDeriv}
\dot{W}_{T_{\mathrm{sh},i}} & \leq -\frac{q_{\mathrm{p},i}}{2V_{\mathrm{sh},i}}\left(T_{\mathrm{sh},i}-T_{\mathrm{p},i}^\star \right)^2 + \frac{q_{\mathrm{p},i}}{2V_{\mathrm{sh},i}}\left(T_{\mathrm{p},i}-T_{\mathrm{p},i}^\star \right)^2.
\end{align}
{\color{black}To verify claim {\bf (I)}, note that $T_{\mathrm{p},i}(t)$ converges to $T_{\mathrm{p},i}^\star$ monotonically by Proposition \ref{prop:producerTemp}. As $|T_{\mathrm{p},i}(0) - T_{\mathrm{p},i}^\star| \leq \phi_i$, the bound $|T_{\mathrm{p},i}(t) - T_{\mathrm{p},i}^\star| \leq \phi_i$ is satisfied all the time. Applying this bound to \eqref{WTshDeriv} results in
\begin{equation}
	\dot{W}_{T_{\mathrm{sh},i}}
	\leq
	-\frac{1}{2V_{\mathrm{sh},i}}q_{\mathrm{p},i}\left(T_{\mathrm{sh},i}-T_{\mathrm{p},i}^\star \right)^2 \nonumber + \frac{1}{2V_{\mathrm{sh},i}}q_{\mathrm{p},i}\phi_i^2,
\end{equation}
 whose right-hand side is non-positive for all $|T_{\mathrm{sh},i}(t) - T_{\mathrm{p},i}^\star| \geq \phi_i$. Consequently, as $|T_{\mathrm{sh},i}(0) - T_{\mathrm{p},i}^\star| \leq \phi_i$, it follows that $|T_{\mathrm{sh},i}(t) - T_{\mathrm{p},i}^\star| \leq \phi_i$ all the time.} {\color{black}Considering claim {\bf (II)}, note that $T_{\mathrm{p},i}-T_{\mathrm{p},i}^\star \to 0$ by Proposition \ref{prop:producerTemp}. Recalling that $q_{\mathrm{p},i}$ is strictly positive, asymptotic stability follows from \eqref{WTshDeriv}.}
\end{proof}

{\color{black}
\begin{remark}
	Propositions \ref{prop:producerTemp} and \ref{prop:vsh_controller} rely on the exact compensation of some system dynamics. It can be shown with extended analysis that the corresponding closed-loop systems are ISS with respect to control imperfection and thus subsequent results share similar properties. This addition analysis, however, is omitted for brevity.
\end{remark}
}

\subsection{Temperature stability of the supply layer}\label{sec:supply layer}
We now focus on the dynamic behaviour of the DN's hot-layer and verify that the temperature of each node and edge within the layer is bounded and remains above a threshold value required for the consumers to operate correctly. Before proceeding, we define $T_{\mathrm{c}_{\mathrm{max}}}^\star$ to be the maximum temperature reference among all consumers. We similarly define $T^\star_{\mathrm{p}_{\mathrm{min}}}$ to be the minimum temperature reference among all producers. It is assumed that $T^\star_{\mathrm{p}_{\mathrm{min}}} > T^\star_{\mathrm{c}_{\mathrm{max}}} + \epsilon$ for some $\epsilon>0$.

\begin{proposition}\label{prop:supplyLayerTemp}
	Consider the set of all temperatures within the distribution supply layer $T_{\mathrm{s}}$ and assume that all temperatures have initial condition satisfying $T_{\mathrm{s},i}(0) \geq T_{\mathrm{c}_{\mathrm{max}}}^\star + \epsilon$, $i\in \mathcal{G}_\mathrm{s}$. If we additionally assume that the initial conditions of each producer and storage tank satisfy
	\begin{equation}\label{prop:supplyLayerTemp:initialCondition}
		\begin{split}
			|T_{\mathrm{p},i}(0) - T_{\mathrm{p},i}^\star|, |T_{\mathrm{sh},i}(0) - T_{\mathrm{p},i}^\star| \leq T^\star_{\mathrm{p}_\mathrm{min}} - T^\star_{\mathrm{c}_\mathrm{max}} - \epsilon,
		\end{split}
	\end{equation}
	the temperature of the distribution layer is bounded and satisfies 
	\begin{equation}\label{supplyTempBound}
		T_{\mathrm{s},i}(t) \geq T_{\mathrm{c}_\mathrm{max}}^\star + \epsilon,~\forall i\in\mathcal{G}_\mathrm{s}
	\end{equation}
	all the time.
\end{proposition}

\begin{proof}
	By Proposition \ref{prop:HotStorageTemperature}, the tank temperatures satisfy $T_{\mathrm{sh},i}(t) \geq T_\mathrm{c_{max}}^\star + \epsilon$ all the time. To verify the behavior of the distribution layer, we consider the coldest temperature in the layer and show that it is lower bounded by $T_\mathrm{c_{max}}^\star + \epsilon$. The temperature dynamics of each edge and node in the distribution layer are described by \eqref{eq:pipe_dyn_DN} and \eqref{eq:node_dyn_DN} with $\chi=s$.
	
	At an arbitrary time $t$, the coldest temperature could occur at either an edge $T_{\mathrm{s},i}$, which would imply that it is colder than all node temperatures, i.e.,  $T_{\mathrm{s},i} \leq T_{\mathrm{s},j}$, $j\in \mathcal{N}_\mathrm{s}$. Recalling \eqref{eq:pipe_dyn_DN} and noting that by the node ordering convention $q_{\mathrm{s},i} \geq 0$, we have that $V_{\mathrm{s},i}  \dot{T}_{\mathrm{s},i} \geq 0$, ensuring that if the coldest temperature is within an edge, it is non-decreasing. Now, consider that the coldest temperature is within a node $T_{\mathrm{s},j}$. As the node is the coldest within the network, it is colder than all edges, i.e.,   $T_{\mathrm{s},j} \leq T_{\mathrm{s},i}$, $i\in \mathcal{E}_\mathrm{s}$. Recalling \eqref{eq:node_dyn_DN} and noting that by the node ordering convention $q_{\mathrm{s},i} \geq 0$, we have that
	\begin{equation}\label{nodeTemp2}
		\begin{split}
			 V_{\mathrm{s},j} \dot{T}_{\mathrm{s},j}  
			 \geq 
			 \sum_{k=1}^{n_\mathrm{p}} \alpha_{k,j}q_{\mathrm{st},k}\left(T_{\mathrm{sh},k}-T_{\mathrm{s},j}\right).
		\end{split}
	\end{equation}
	As each storage tank satisfies $T_{\mathrm{sh},i}(t) \geq T_\mathrm{c_{max}}^\star + \epsilon$  all the time, $\dot{T}_{\mathrm{s},j}$ is non-decreasing for $T_{\mathrm{s},j} \leq T_\mathrm{c_{max}}^\star + \epsilon$. As this inequality holds all the time and all temperatures initially satisfy $T_{\mathrm{s},j} \geq T_\mathrm{c_{max}}^\star + \epsilon$ it follows that all temperatures within the supply layer are lower bounded by $T_\mathrm{c_{max}}^\star + \epsilon$.
	
	By a similar argument, we can conclude that the distribution layer is upper-bounded as a function of the layer temperature initial conditions $T_{\mathrm{s},i}, T_{\mathrm{s},j}$, storage tank initial conditions $T_{\mathrm{sh},k}(0)$ and producer reference set-points $T_{\mathrm{p},i}^\star$. %We conclude that, under the stated assumptions, all temperatures within the supply layer are bounded and lower-bounded by $T_{c_{max}}^\star$.
\end{proof}

{\color{black}The assumptions in Proposition \ref{prop:supplyLayerTemp} require that the  initial temperature of the supply layer is higher than $T_{\mathrm{c}_{\max}}^\star$ and that the initial temperature of each producer and storage tank hot layer belongs to a boundary layer centered in $T_{\mathrm{p},i}^\star$ (see \eqref{prop:supplyLayerTemp:initialCondition}).}

\subsection{Control of consumer temperatures}

Now we propose a simple control law  to ensure the regulation of the $i$th consumer's return temperature $T_{\mathrm{c},i}$  to some specified value $T_{\mathrm{c},i}^\star$. It is assumed that the consumer can measure the temperature $T_{\mathrm{c},i}^{\mathrm{in}}$ of the incoming stream from the DN's supply layer (see Fig.~\ref{fig:hydraulic_complex_elements} and \eqref{eq:temp inlet storage and consumers}). This value, however, does not need to be constant.
%Moreover, it is considered that the consumer remains off-line if the incoming water's temperature is below some known constant  value, which is chosen to be equal to $T_{\mathrm{c},i}^\star$. Once the incoming stream's temperature satisfies
%\begin{equation}\label{eq: condition online consumer}
%T_{\mathrm{c},i}^\mathrm{in}-T_{\mathrm{c},i}^\star\geq \epsilon_{\mathrm{c},i},
%\end{equation}
%for some $\epsilon_{\mathrm{c},i}>0$, the consumer comes online. 
Before proceeding with the control design, it is recalled that $P_{\mathrm{c},i}\geq 0$ is constant and unknown. %Also, the temperature of stream at the inlet ($T_{\mathrm{c},i}^\mathrm{in}$) and the outlet ($T_{\mathrm{c},i}$) of the $i$th consumer are available for control design. Nonetheless, it is noted that $T_{\mathrm{c},i}^\mathrm{in}$ is not required to be constant.  

\begin{proposition}\label{prop:consumer}
Consider the $i$th consumer's temperature dynamics \eqref{eq:consumer dynamics} in closed-loop with the control law
\begin{equation}\label{eq: control law q_c (a)}
q_{\mathrm{c},i}(T_{\mathrm{c},i}^\mathrm{in}, z_{\mathrm{c},i})=\frac{1}{T_{\mathrm{c},i}^\mathrm{in}-T_{\mathrm{c},i}^\star}z_{c,i},
\end{equation}
with initial condition {\color{black}$z_{\mathrm{c},i}(0) > 0$} where
\begin{subequations}\label{eq: control law q_c (b)}
\begin{align}
z_{\mathrm{c},i} & =  x_{\mathrm{c},i}-V_{\mathrm{c},i}T_{\mathrm{c},i}, \label{eq: control law q_c (c)} \\
\dot{x}_{\mathrm{c},i} & = q_{\mathrm{c},i}\left(T_{\mathrm{c},i}^\mathrm{in}-T_{\mathrm{c},i} \right)-z_{\mathrm{c},i}. \label{eq: control law q_c (d)}
\end{align}
\end{subequations}
The resulting closed-loop dynamics are such that $T_{\mathrm{c},i}\rightarrow T_{\mathrm{c},i}^\star$ and $z_{\mathrm{c},i}\rightarrow P_{\mathrm{c},i}^\star$ exponentially (as $t\rightarrow \infty$). Also, the state $z_{\mathrm{c},i}$, and hence the input $q_{\mathrm{c},i}$, are strictly non-negative. 
\end{proposition}

\begin{proof}
{\color{black}Taking the time derivative of \eqref{eq: control law q_c (c)} and substituting in \eqref{eq:consumer dynamics}, the dynamics of $z_{\mathrm{c},i}$ can be written as
}
\begin{align}\label{eq: proof T_c dyn CL (a)}
\dot{z}_{\mathrm{c},i} 
%& =\dot{x}_{\mathrm{c},i}-V_{\mathrm{c},i}\dot{T}_{\mathrm{c},i} \nonumber \\
%& = q_{\mathrm{c},i}\left(T_{\mathrm{c},i}^\mathrm{in}-T_{\mathrm{c},i} \right)-z_{\mathrm{c},i}-q_{\mathrm{c},i}\left( T_{\mathrm{c},i}^\mathrm{in}-T_{\mathrm{c},i} \right)+P_{\mathrm{c},i} \nonumber \\
& = -z_{\mathrm{c},i}+P_{\mathrm{c},i},
\end{align}
which ensures that $z_{\mathrm{c},i}$ is strictly positive for any $P_{\mathrm{c},i}> 0$ and converges to $P_{\mathrm{c},i}$ exponentially. Recalling the definition of $T_{\mathrm{c},i}^\mathrm{in}$ in \eqref{eq:temp inlet storage and consumers}, Proposition \ref{prop:supplyLayerTemp} implies that
\begin{equation}\label{epsilonCondition}
	T_{\mathrm{c},i}^\mathrm{in}-T_{\mathrm{c},i}^\star \geq \epsilon > 0,
\end{equation}
resulting in $q_{\mathrm{c},i} \geq 0$  all the time.

By substituting the control law \eqref{eq: control law q_c (a)}, \eqref{eq: control law q_c (b)} into the consumer's temperature dynamics \eqref{eq:consumer dynamics} results in:
\begin{align}\label{eq: proof T_c dyn CL (b)}
 V_{\mathrm{c},i}\dot{T}_{\mathrm{c},i}
%  = q_{\mathrm{c},i}\left(T_{\mathrm{c},i}^\mathrm{in}-T_{\mathrm{c},i} \right)-P_{\mathrm{c},i} \nonumber \\
%& = -\left(T_{\mathrm{c},i}-T_{\mathrm{c},i} \right)q_{\mathrm{c},i}+q_{\mathrm{c},i}\left(T_{\mathrm{c},i}^\mathrm{in}-T_{\mathrm{c},i} \right)-P_{\mathrm{c},i} \nonumber \\
& = -\frac{\left(T_{\mathrm{c},i}-T_{\mathrm{c},i}^\star \right)}{\left( T_{\mathrm{c},i}^\mathrm{in}-T_{\mathrm{c},i}^\star  \right)}z_{\mathrm{c},i}+\left(z_{\mathrm{c},i}-P_{\mathrm{c},i} \right).
\end{align}

To verify stability of the consumer's closed-loop temperature dynamics, consider the Lyapunov candidate
\begin{equation}\label{consumerLyap}
W_{T_{\mathrm{c},i}}= \frac{1}{2}\begin{bmatrix}
T_{\mathrm{c},i}-T_{\mathrm{c},i}^\star\\
z_{\mathrm{c},i}-P_{\mathrm{c},i}
\end{bmatrix}^\top\begin{bmatrix}
V_{\mathrm{c},i} & 1\\
1 & \omega_{\mathrm{c},i}
\end{bmatrix}\begin{bmatrix}
T_{\mathrm{c},i}-T_{\mathrm{c},i}^\star\\
z_{\mathrm{c},i}-P_{\mathrm{c},i}
\end{bmatrix},
\end{equation}
where $\omega_{\mathrm{c},i}$ is a positive constant satisfying $\omega_{\mathrm{c},i} > \frac{1}{V_{\mathrm{c},i}}$ to ensure positivity of $W_{T_{\mathrm{c},i}}$.
%The time derivative of $W_{T_{\mathrm{c},i}}$ along the system's solutions satisfies
%\begin{align}
%\dot{W}_{T_{\mathrm{c},i}}=\begin{bmatrix}
%T_{\mathrm{c},i}-T_{\mathrm{c},i}^\star\\
%z_{\mathrm{c},i}-P_{\mathrm{c},i}
%\end{bmatrix}^\top\begin{bmatrix}
%V_{\mathrm{c},i} & 1\\
%1 & \omega_{\mathrm{c},i}
%\end{bmatrix}\begin{bmatrix}
%\frac{1}{V_{\mathrm{c},i}}V_{\mathrm{c},i}\dot{T}_{\mathrm{c},i}\\
%\dot{z}_{\mathrm{c},i}
%\end{bmatrix}.
%\end{align}
Considering  \eqref{eq: proof T_c dyn CL (a)} and \eqref{eq: proof T_c dyn CL (b)}, it can be  verified (through lengthy, yet direct computations), that
\begin{equation}
\dot{W}_{T_{\mathrm{c},i}}=-\begin{bmatrix}
T_{\mathrm{c},i}-T_{\mathrm{c},i}^\star\\
z_{\mathrm{c},i}-P_{\mathrm{c},i}
\end{bmatrix}^\top 
\mathcal{R}_{T_{\mathrm{c},i}}\begin{bmatrix}
T_{\mathrm{c},i}-T_{\mathrm{c},i}^\star\\
z_{\mathrm{c},i}-P_{\mathrm{c},i}
\end{bmatrix}^\top,
\end{equation}
where
\begin{equation}
\mathcal{R}_{T_{\mathrm{c},i}}=\begin{bmatrix}
\frac{1}{T_{\mathrm{c},i}^\mathrm{in}-T_{\mathrm{c},i}^\star}z_{\mathrm{c},i} & \frac{1}{2V_{\mathrm{c},i}}\frac{1}{T_{\mathrm{c},i}^\mathrm{in}-T_{\mathrm{c},i}^\star}z_{\mathrm{c},i}\\
 \frac{1}{2V_{\mathrm{c},i}}\frac{1}{T_{\mathrm{c},i}^\mathrm{in}-T_{\mathrm{c},i}^\star}z_{\mathrm{c},i} & \omega_{\mathrm{c},i}-\frac{1}{V_{\mathrm{c},i}}
\end{bmatrix}.
\end{equation}
From the Schur complement condition for positive semi-definiteness, the matrix $\mathcal{R}_{T_{\mathrm{c},i}}$ is positive-definite provided that  $\omega_{\mathrm{c},i}$ satisfies
\begin{equation}\label{eq: proof T_c dyn CL (c)}
\omega_{\mathrm{c},i} > \frac{1}{V_{\mathrm{c},i}}+ \frac{1}{4V_{\mathrm{c},i}^2}\frac{1}{T_{\mathrm{c},i}^\mathrm{in}-T_{\mathrm{c},i}^\star}z_{\mathrm{c},i}.
\end{equation}
Note from \eqref{eq: proof T_c dyn CL (a)} that $z_{\mathrm{c},i}$ converges to $P_{\mathrm{c},i}$ exponentially, which implies that it has a finite upper bound $\Vert z_{\mathrm{c},i} \Vert_\infty$. Applying the condition \eqref{epsilonCondition}, \eqref{eq: proof T_c dyn CL (c)} holds  for any $\omega_{\mathrm{c},i}$ satisfying
\begin{equation}\label{eq: proof T_c dyn CL (d)}
\omega_{\mathrm{c},i}>\frac{1}{V_{\mathrm{c},i}}+\frac{1}{4V_{\mathrm{c},i}^2}\frac{1}{\epsilon}\Vert z_{\mathrm{c},i} \Vert_\infty.
\end{equation}
Note that the right-hand side of \eqref{eq: proof T_c dyn CL (d)} is greater or equal than that of \eqref{eq: proof T_c dyn CL (c)} and ensures positivity of \eqref{consumerLyap}. Since $\omega_{\mathrm{c},i}$ is not used in the definition of the control law, the closed-loop system \eqref{eq:consumer dynamics}, \eqref{eq: control law q_c (a)}, \eqref{eq: control law q_c (b)}  is stable for any $\omega_{\mathrm{c},i}$ satisfying~\eqref{eq: proof T_c dyn CL (d)}.
\end{proof}

\begin{remark}
Equations \eqref{eq: controller q_p} and \eqref{eq: control law q_c (a)} define control laws for each $q_{\mathrm{p},i}$ and $q_{\mathrm{c},i}$, respectively.   The remaining flow inputs, namely $q_{\mathrm{st},i}$ as well as $q_{\mathrm{s},i}$ ($q_{\mathrm{r},i}$) associated to chords of $\mathcal{G}_\mathrm{s}$ ($\mathcal{G}_\mathrm{r}$)---see Remark~\ref{rem:independent_flows}---can be fixed for simplicity, and to comply with Assumption~\ref{assu:positivity_flow_inputs},   as positive constants. Nonetheless,  the PI-like control law reported in \cite[Sec.~4.1]{Scholten2017} can also be used, which is non-negative all the time.
\end{remark}
{\color{black}
\begin{remark}
In this work, the consumer's power consumption $P_{c,i}$ is assumed to be constant. In practice, however, this quantity would be time-varying. If the dynamics of the power consumption are slow compared to the closed-loop network then changes in the load can be neglected. If the loads are periodic the loads can be dynamically compensated using a frequency estimator, similar to \cite{Gentili2007}.
\end{remark}
}
\subsection{Stability of the cold layer}\label{subsec:stab_cold_lay}
In this subsection, we study the stability of the cold layer, consisting of the return layer temperature and cold storage tank temperatures and volumes. The analysis is analogous to the one performed to assess the stability of the DH system's hot layer in Section \ref{sec:supply layer}.

\begin{proposition}\label{prop:returnLayerTemp}
	All edge and node temperatures within the return layer ($T_{\mathrm{r},i}$, $i\in \mathcal{G}_\mathrm{r}$) are bounded  all the time.
\end{proposition}

\begin{proof}
	From Proposition \ref{prop:consumer}, all consumer temperatures are bounded and converge to $T_{\mathrm{c},i}^\star$. The remainder of the proof follows analogously to the proof of Proposition \ref{prop:supplyLayerTemp}.
\end{proof}

\begin{proposition}\label{prop:coldStorageTemp}
	The volume of each cold-layer storage tank $V_{\mathrm{sc},i}$ is Lyapunov stable and asymptotically stable for $q_{\mathrm{st},i} > 0$. The temperature is bounded all the time.
\end{proposition}

\begin{proof}
	Considering the reciprocal behaviour of the hot and cold-layer storage tanks in \eqref{eq:producer+storage dyn (c)}, \eqref{eq:producer+storage dyn (d)} and Proposition \ref{prop:vsh_controller}, the volume of the cold storage tank $V_{\mathrm{sc},i}$ shares the same stability properties of the hot-layer storage $V_{\mathrm{sh},i}$. The cold storage tank temperature is described by \eqref{eq:cold storage dyn}. Proposition \ref{prop:returnLayerTemp} ensures that temperatures within the return layer are bounded, which ensures that the temperature of the cold storage tank is bounded as well.
\end{proof}

%\subsection{Storage tanks' temperature stability (cold layer)}
%
%As indicated in Remark~\ref{rem:q_p for iss and stability v_sc}, the volume $V_{\mathrm{sc},i}$ has a stable behavior as a consequence of Proposition~\ref{prop:vsh_controller}. Next we verify that the temperature of the cold layer of the $i$th storage tank behaves stably provided that the temperature of the stream at the inlet behaves stably as well.
%
%{\color{red} This result needs to be adjusted with respect to Joel's note V04. Here the cold layer of each storage tank takes water from the return layer of the distribution network. Based on the latest developments, at best we can establish that the distribution networks will have a stable behavior but cannot ensure convergence. Maybe we should analyze the behavior of $T_{\mathrm{sc},i}$ considering two situations. One, either $T_{\mathrm{sc},i}^\mathrm{in}$ is bounded and eventually converges to some equilibrium. Two, $T_{\mathrm{sc},i}^\mathrm{in}$ is only stable, but does not describe asymptotic behavior. In the first case we can use the proof for the stability of $T_{\mathrm{sc},i}$ as it is in Joel's note V04. For the second case we should focus on establishing that the $T_{\mathrm{sc},i}$-dynamics is input-to-state stable with respect to $T_{\mathrm{sc},i}^\mathrm{in}$. Is it feasible to establish the second case?}

\subsection{Overall system stability}
The final property to be verified is stability of the overall DH system. To achieve this, we invoke the propositions detailed through subsections~\ref{subsec:producer_temp}--\ref{subsec:stab_cold_lay} and verify that the assumptions utilized throughout are satisfied by the closed-loop system.

\begin{theorem}
	Consider the DH model detailed in Section~\ref{sec:model} in closed-loop with the decentralized control scheme described through subsections~\ref{subsec:producer_temp}--\ref{subsec:stab_cold_lay}. Assuming that Assumption \ref{assu:positivity_flow_inputs} holds, the producer and storage tank initial temperatures satisfy \eqref{prop:supplyLayerTemp:initialCondition} and the supply layer initial temperatures satisfy \eqref{supplyTempBound}, the closed-loop system satisfies the following properties: {\bf (I)}  The producer temperatures $T_{\mathrm{p},i}$ and hot storage temperatures $T_{\mathrm{sh},i}$ converge to the producer reference temperature $T_{\mathrm{p},i}^\star$. {\bf (II)} The consumer temperatures $T_{\mathrm{c},i}$ converge their reference temperatures $T_{\mathrm{c},i}^\star$. {\bf (III)} All network temperatures remain bounded. {\bf (IV)} The hot-layer storage tank volumes $V_{\mathrm{sh},i}$ converge to the reference value $V_{\mathrm{sh},i}^\star$ and the cold-layer storage tank volumes $V_{\mathrm{sc},i}$ converge to a constant value.
%	\begin{enumerate}
%	\item\label{thm:producer} The producer temperatures $T_{r,i}$ and hot storage temperatures $T_{sh,i}$ converge to the producer reference temperature $T_{pr,i}^\star$.
%	
%	\item\label{thm:consumer} The consumer temperatures $T_{c,i}$ converge their reference temperatures $T_{c,i}^\star$.
%	
%	\item\label{thm:temps} All network temperatures remain bounded.
%	
%	\item\label{thm:volume} The hot-layer storage tank volumes $V_{sh,i}$ converge to the reference value $V_{sh,i}^\star$ and the cold-layer storage tank volumes $V_{sc,i}$ converge.
%	\end{enumerate}
\end{theorem}

\begin{proof}
	Claim {\bf (I)} follows from direct application of Propositions \ref{prop:producerTemp} and \ref{prop:HotStorageTemperature}.
	The temperatures of the distribution layer satisfy the lower bound \eqref{supplyTempBound} by Proposition \ref{prop:supplyLayerTemp} which ensures that the inequality \eqref{supplyTempBound} is satisfied. Consequently, claim {\bf (II)} is satisfied by Proposition \ref{prop:consumer}.
	Claim {\bf (III)} follows from direct application of Propositions \ref{prop:supplyLayerTemp}, \ref{prop:returnLayerTemp} and \ref{prop:coldStorageTemp}. Finally claim {\bf (IV)} follows from application of Propositions \ref{prop:vsh_controller} and \ref{prop:coldStorageTemp}.
\end{proof}

\section{Numerical Simulations}\label{sec:num_sims}

In this section the performance of the DH system model in closed-loop with the proposed controllers is illustrated via numerical simulations. The configuration and data are based on the case study reported in the arXiv version of \cite[Section~4]{jmc_dh_modeling_2020}, which corresponds to a DH system with three heat producers ($n_\mathrm{p}=3$),  nine consumers ($n_c=9$) and with the same topology as the sketch shown in Fig.~\ref{fig:1}.  Each tank is assumed to have a total capacity of $1000~\mathrm{m}^3$.

\begin{figure}[t]
\begin{center}
\includegraphics[width=0.8\linewidth]{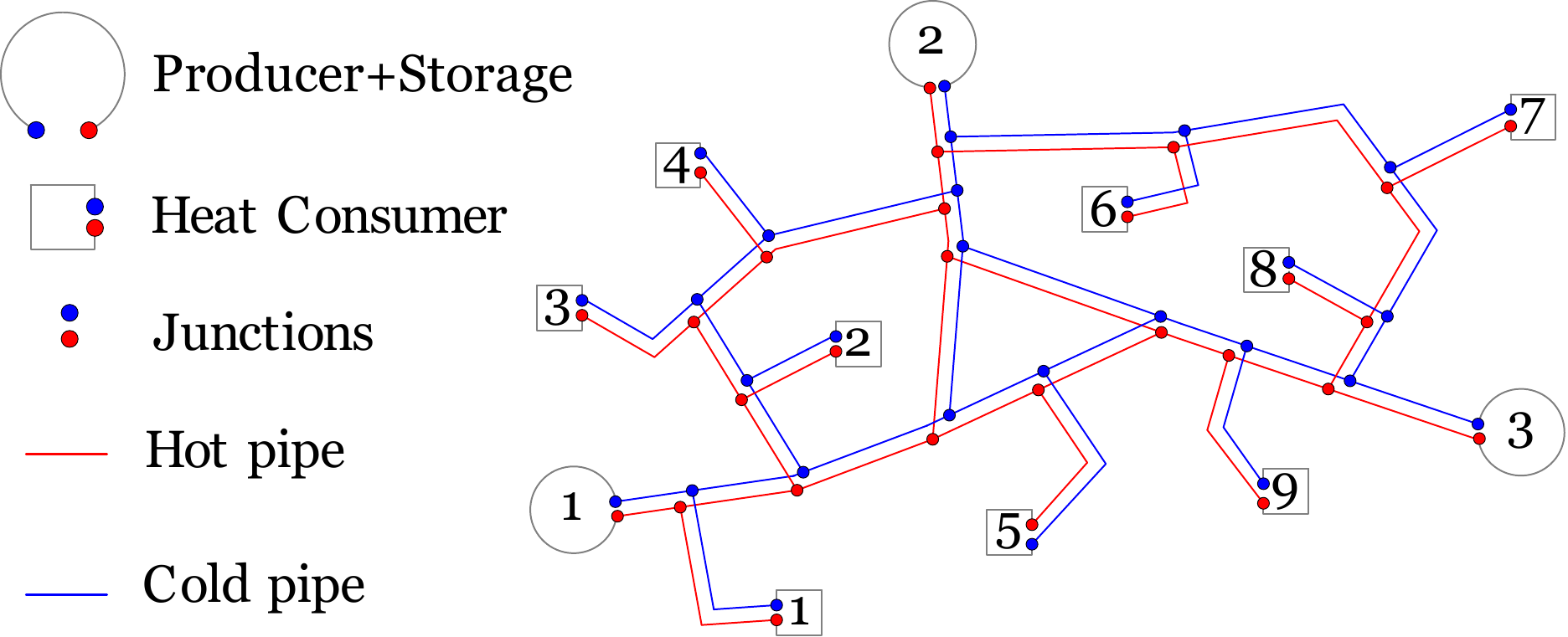}
\caption{Sketch of a simplified DH system (c.f.,  \cite{wang_meshed_17}).}
\label{fig:1}
\end{center}
\end{figure}

%We test the closed-loop system  by setting conditions associated to low and high heat demand from consumers, and charging and discharging processes for the storage tanks. For simplicity we assume that consumers' heat demand is proportional to their flow setpoints. Then,  low consumer demand is characterized by choosing small entries for $q_\mathrm{ch}^\star$ with respect to a given nominal flow. Conversely, high consumer demand is associated with large values for the entries of $q_\mathrm{ch}^\star$.  
%For the storage tanks, they are considered to be in charging mode if the entries of $V_\mathrm{sh}^\star$ are switched from small to large values relative  to the tanks' total capacity. An analogous, yet converse description follows for the tanks' discharging mode.

The tuning gains $k_{\mathrm{p},i}$ of the controllers \eqref{eq:controller P_pi} are all equal to $1\times 10^{-3}$. The same values are chosen for the gains  $\kappa_{\mathrm{p},i}$ of the controllers $q_{\mathrm{p},i}$ in \eqref{eq: controller q_p}. This selection is based on a trial-and-error procedure aimed at attaining a fair balance between settling time and overshoot for the signals of interest. The producers' and consumers' temperature setpoints $T_{\mathrm{p},i}^\star$ and $T_{\mathrm{c},i}^\star$ are chosen to be equal to $85~^\circ\mathrm{C}$ and $55~^\circ\mathrm{C}$, respectively.

%\begin{figure}[ht]
%\begin{subfigure}{\linewidth}
%\includegraphics[width=\textwidth]{plot_T_p}
%\end{subfigure}
%%
%\begin{subfigure}{\linewidth}
%\includegraphics[width=\textwidth]{plot_T_c}
%\end{subfigure}
%%
%\begin{subfigure}{\linewidth}
%\includegraphics[width=\textwidth]{plot_V_sh}
%\end{subfigure}
%\caption{Evolution of the flow vector $q_\mathrm{ch}$ (top) and of the volume of hot water in the storage tanks (bottom).}
%\label{fig:sims_1}
%\end{figure}

An  explanation of the simulation results  shown in Fig.~\ref{fig:sims_1} is as follows. The system is initialized in the vicinity of a system's equilibrium in a context of low consumer demand  (50\% w.r.t. full demand) with $P_\mathrm{c}=(1.65,1.98,2.97,2.31,1.98,1.32,3.3,2.31,1.65)~$MW and with relatively small values  for the storage tanks' volumes setpoints with   $V_{\mathrm{sh}}^\star=(100,150,200)~\mathrm{m}^3$. Convergence of the signals on display is observed after a short transient (states and inputs). At $t=6$h all storage tanks switch to a charging mode and attain their respective new  setpoint $V_\mathrm{sh}^\star=(850,900,950)~\mathrm{m}^3$ at approximately $t=9$h. Note that the tanks switching to a charging mode causes  an increase in the producers'  powers  $P_{\mathrm{p},i}$ and flow $q_{\mathrm{p},i}$ during the process. At $t=12$h the consumers' heat demands are simultaneously increased to higher values (75\% w.r.t. full demand). The plot of some consumers' flows  $q_{{c},i}$ shows an increase at this instant and, after a short transient, the consumers' temperatures return to their fixed setpoints. Observe  that the increase in $P_{\mathrm{c},i}$ induces adjustments to the equilibrium values of $P_{\mathrm{p},i}$ and $q_{\mathrm{p},i}$ too, but without significant overshoots. At $t=18$h the tanks switch now to a discharging mode that ends at  approximately $t=21$h. During this process, it is possible to see a reduction in producers' powers and flows, contrary to what is observed during the tanks' charging mode. The new, lower values for the entries of $V_\mathrm{sh}^\star$ are maintained until the end of the simulation. Finally, it is observed that the producers' temperatures $T_{\mathrm{p},i}$ reach quickly (exponentially) their setpoints and remain at this value during the whole simulation time. 

{\color{black} In order to test the effect of potential physical or practical constraints, we performed additional numerical simulations where we saturated the values of all the inputs $q_{\mathrm{p},i}$, $P_{\mathrm{p},i}$ and $q_{\mathrm{c},i}$ such that they are positive and upper bounded by maximum nominal values. In  Fig.~\ref{fig:sims_sat_exp1} we have simulated the same scenario as for the results in Fig.~\ref{fig:sims_1}, considering additionally input constraints as described above. Note that, except for a slower charging rate for the hot layer of storage tank 2, there are no significant behavioral changes of the signals of interest with respect to those in Fig.~\ref{fig:sims_1}. Various additional numerical experiments were conducted by taking random initial conditions with at most a 25\% deviation with respect to a nominal equilibrium point and we obtained similar results (see Fig.~\ref{fig:sims_sat_exp2}).}

\begin{figure*}
\centering
\begin{subfigure}{0.49\textwidth}
\includegraphics[width=\linewidth]{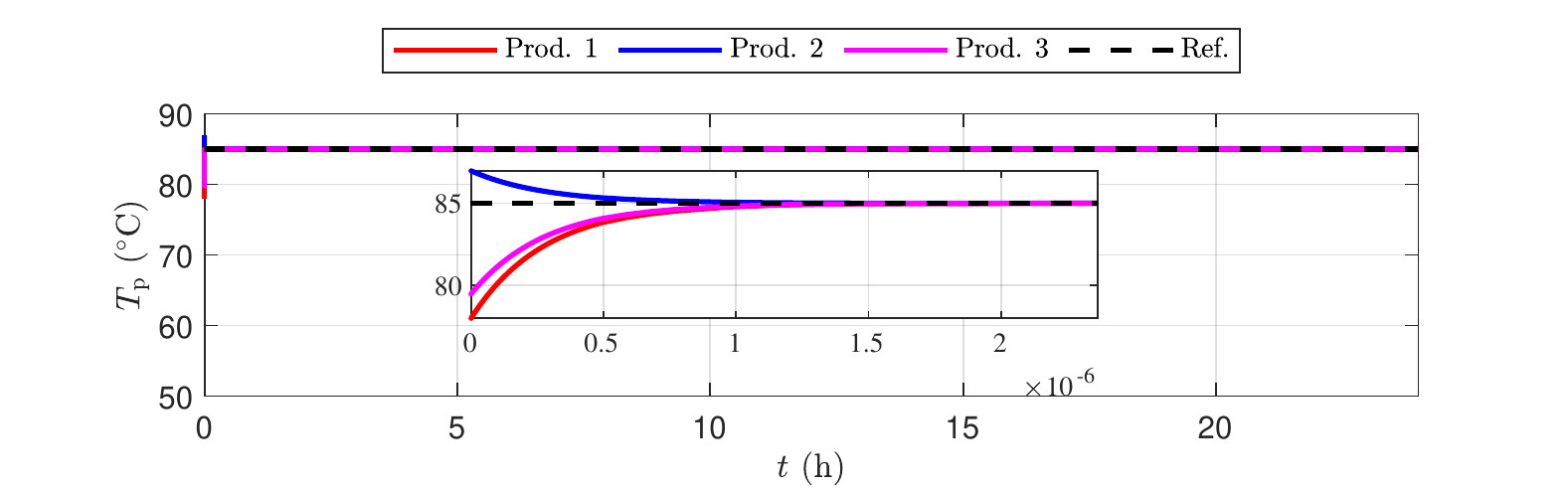}
\caption{~}
\end{subfigure}
\begin{subfigure}{0.49\textwidth}
\includegraphics[width=\linewidth]{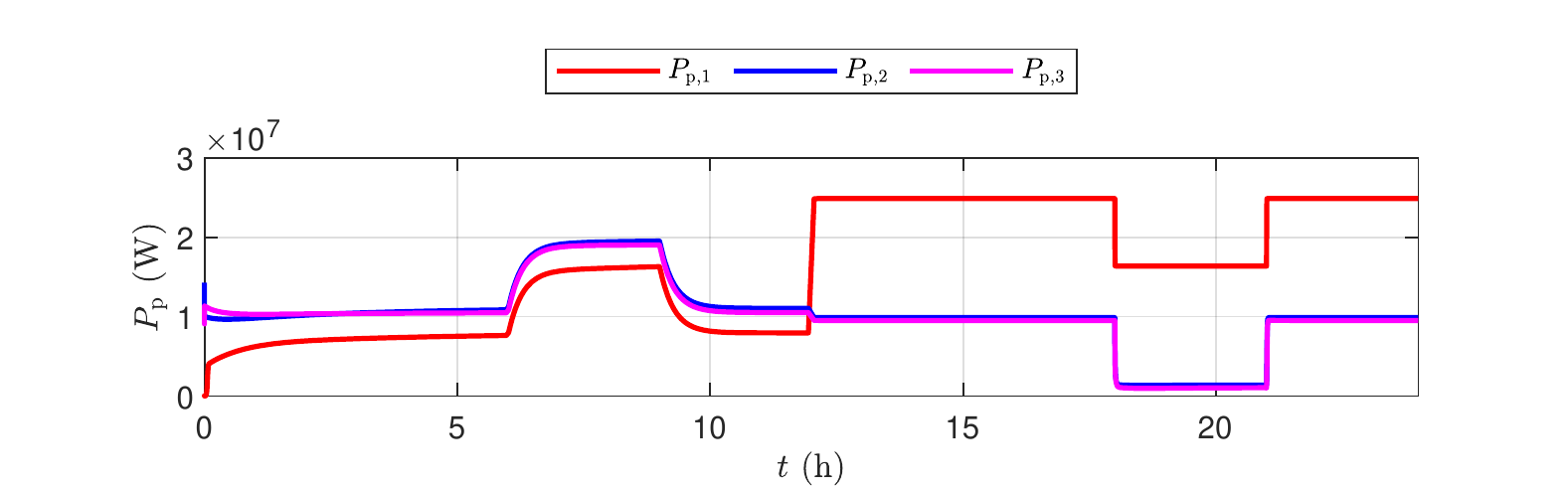}
\caption{~}
\end{subfigure}
\begin{subfigure}{0.49\textwidth}
\includegraphics[width=\linewidth]{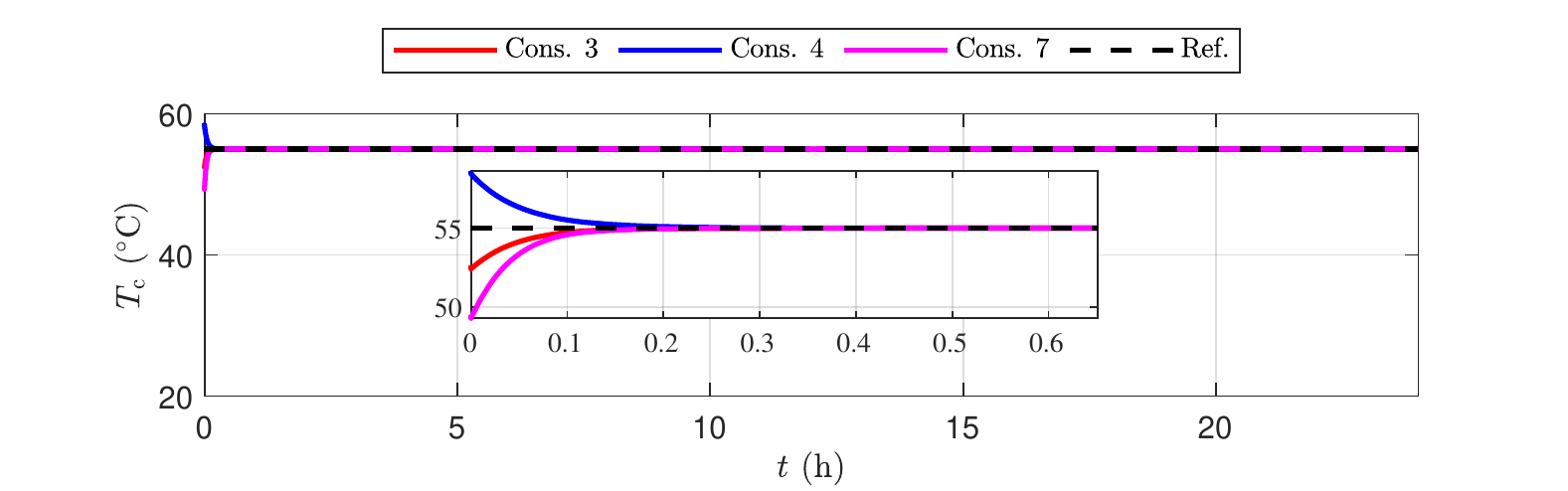}
\caption{~}
\end{subfigure}
\begin{subfigure}{0.49\textwidth}
\includegraphics[width=\linewidth]{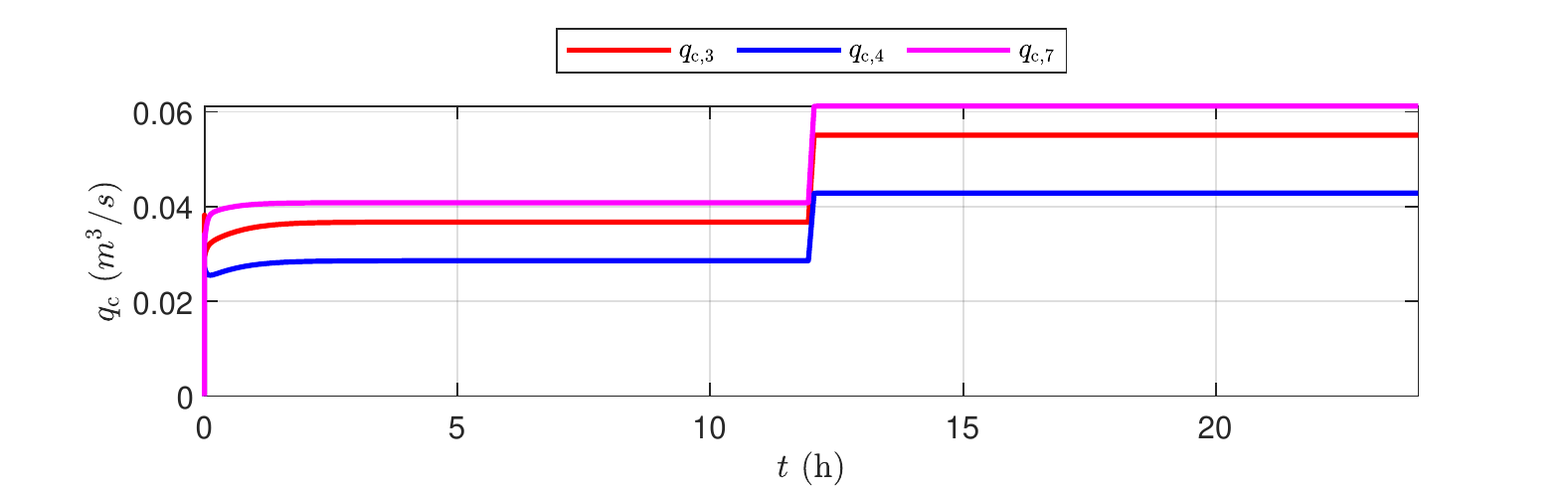}
\caption{~}
\end{subfigure}
\begin{subfigure}{0.49\textwidth}
\includegraphics[width=\linewidth]{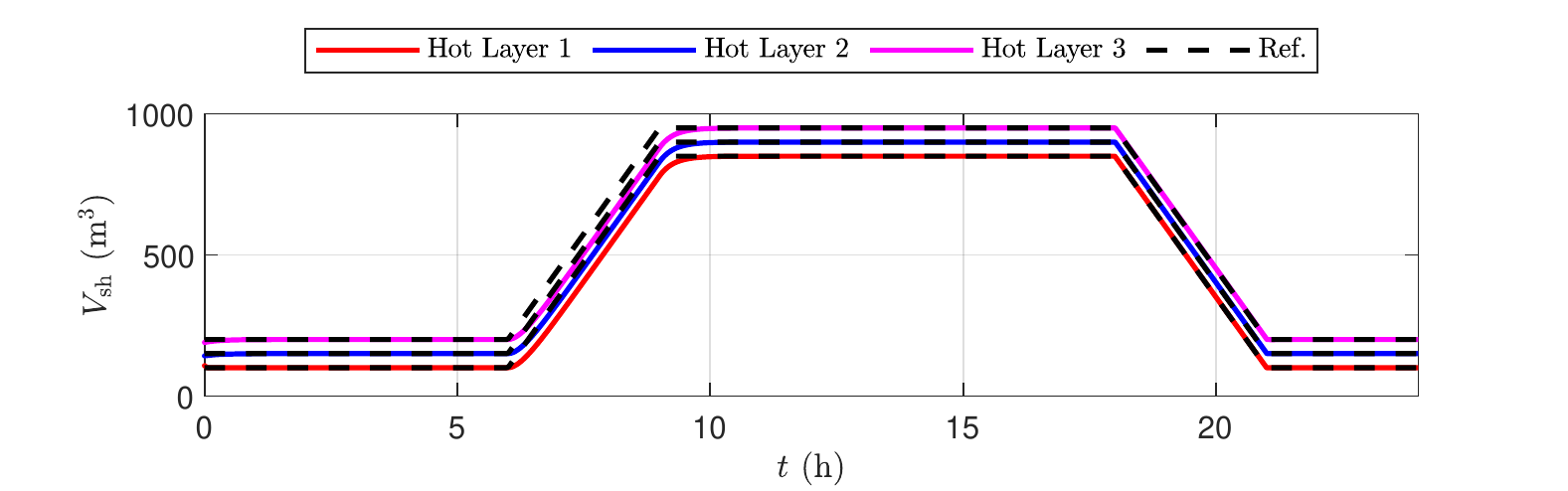}
\caption{~}
\end{subfigure}
\begin{subfigure}{0.49\textwidth}
\includegraphics[width=\linewidth]{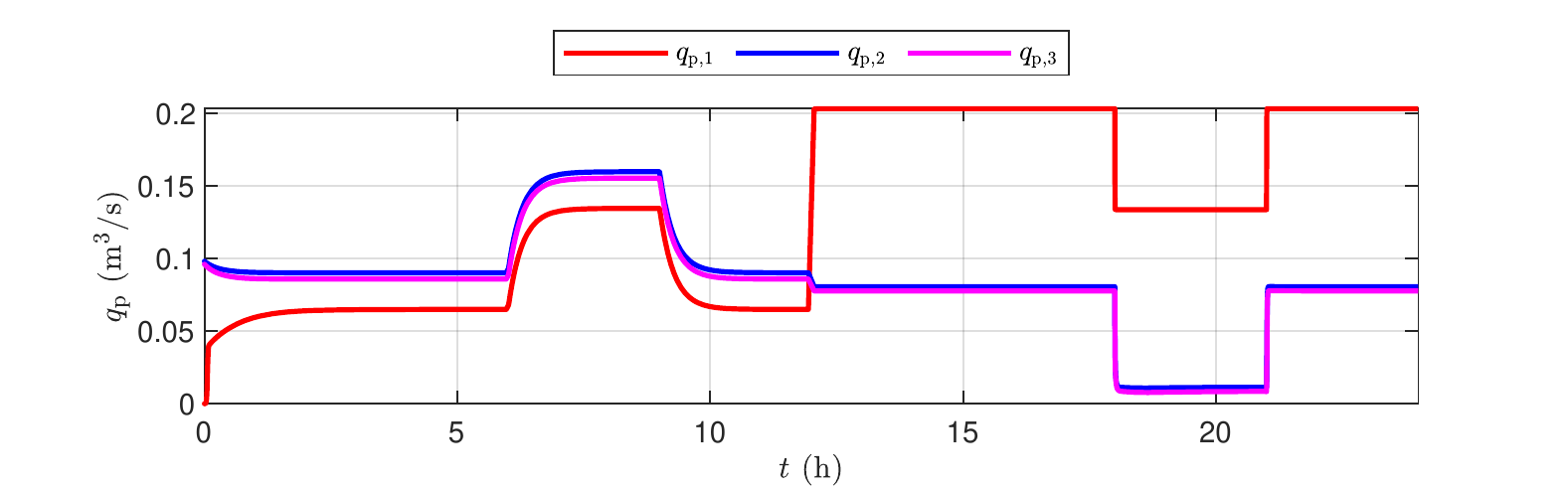}
\caption{~}
\end{subfigure}
\caption{{\bf Left column, top to bottom:} evolution of the producers' temperatures $T_{\mathrm{p},i}$, some of  consumers' outlet temperatures $T_{\mathrm{c},i}$ and the   volume of hot water in the storage tanks $V_{\mathrm{sh},i}$. {\bf Right column, top to bottom:} evolution of the inputs $P_{\mathrm{p},i}$, some of the inputs $q_{\mathrm{c},i}$, and the inputs $q_{\mathrm{p},i}$.}
\label{fig:sims_1}
\end{figure*}

\begin{figure*}
\centering
\begin{subfigure}{0.49\textwidth}
\includegraphics[width=\linewidth]{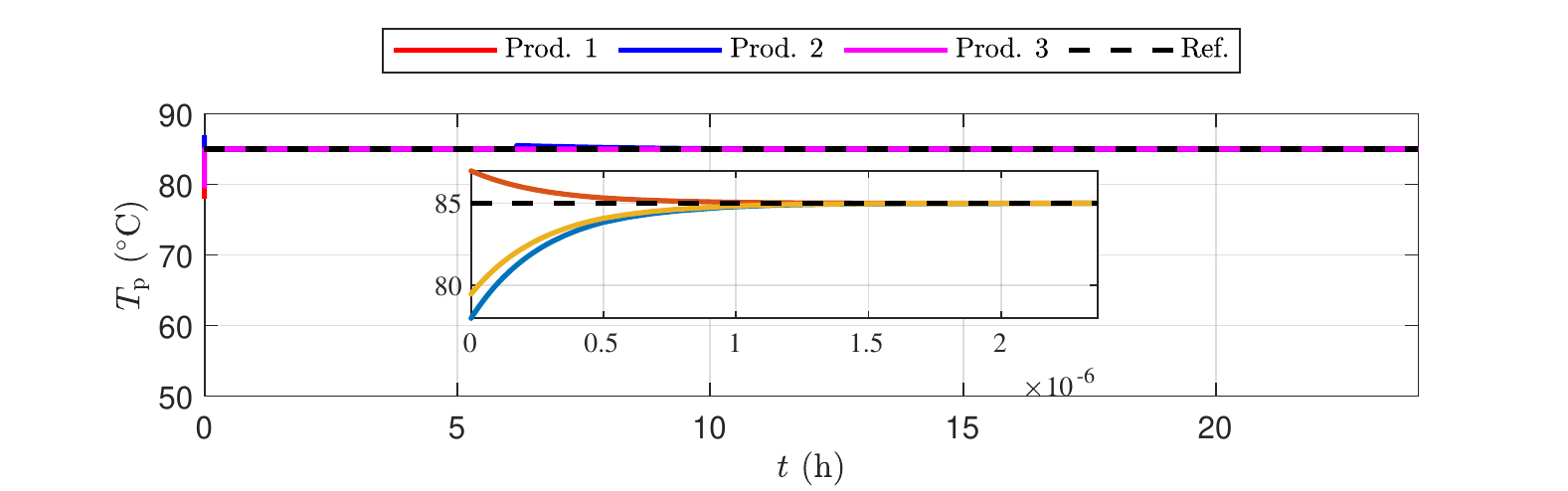}
\caption{~}
\end{subfigure}
\begin{subfigure}{0.49\textwidth}
\includegraphics[width=\linewidth]{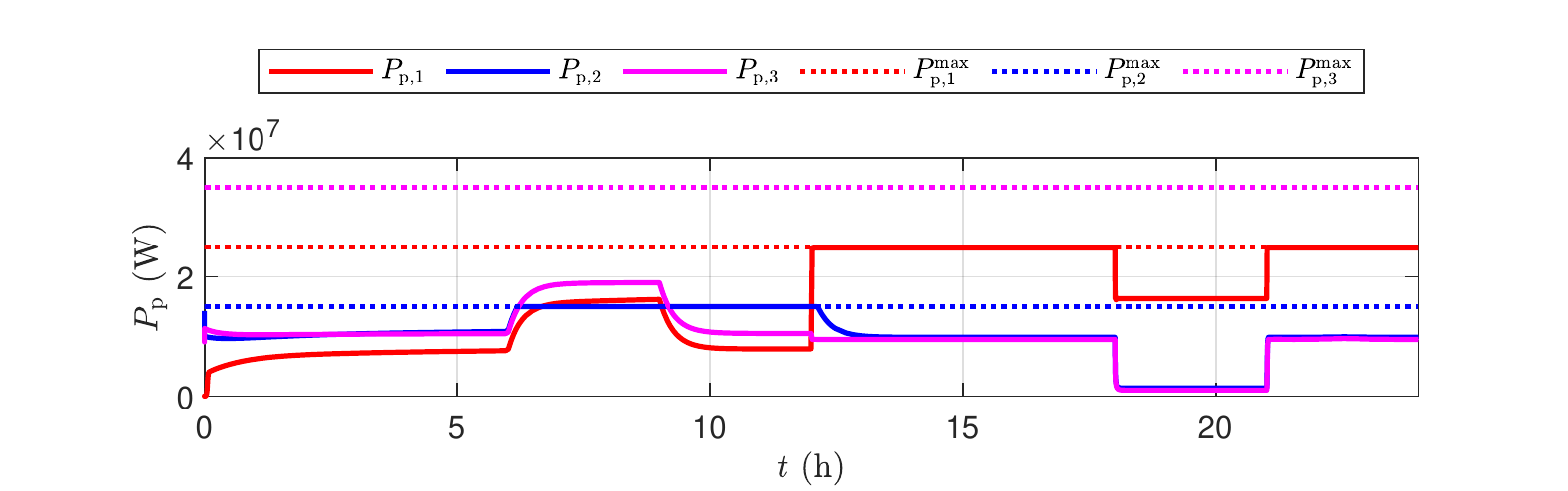}
\caption{~}
\end{subfigure}
\begin{subfigure}{0.49\textwidth}
\includegraphics[width=\linewidth]{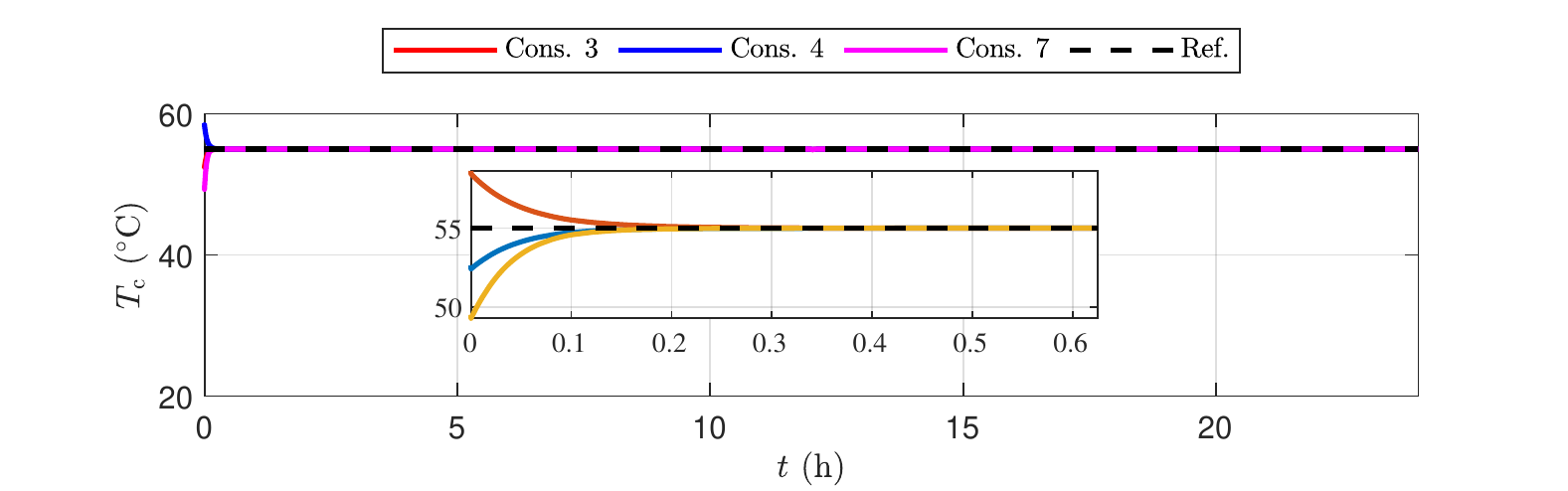}
\caption{~}
\end{subfigure}
\begin{subfigure}{0.49\textwidth}
\includegraphics[width=\linewidth]{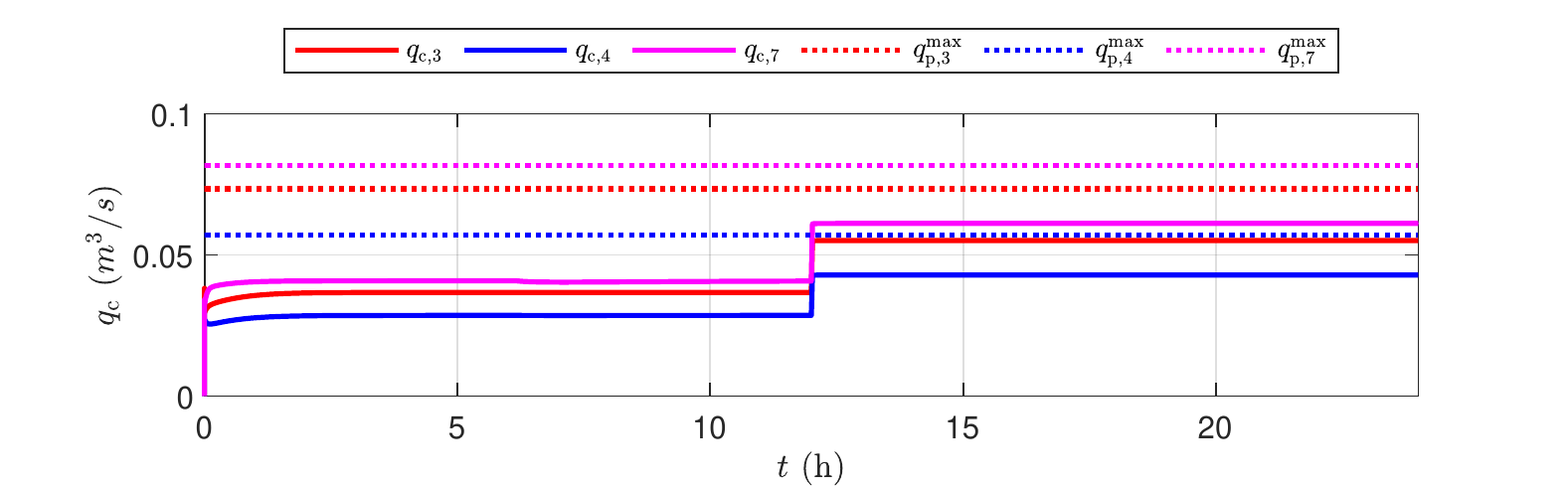}
\caption{~}
\end{subfigure}
\begin{subfigure}{0.49\textwidth}
\includegraphics[width=\linewidth]{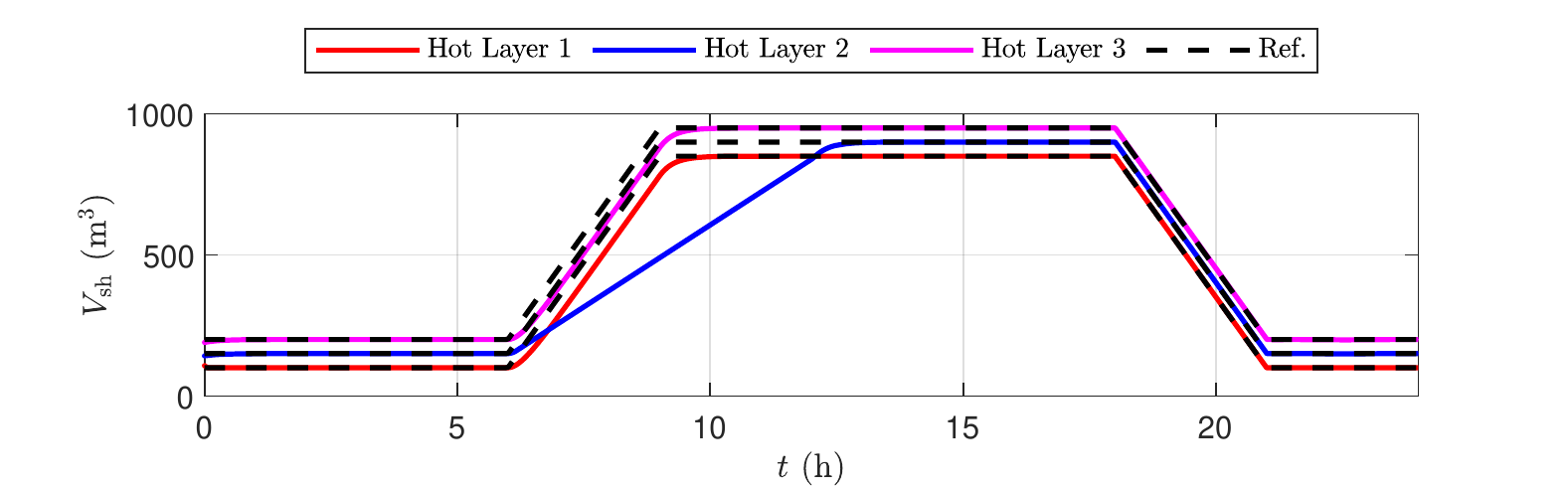}
\caption{~}
\end{subfigure}
\begin{subfigure}{0.49\textwidth}
\includegraphics[width=\linewidth]{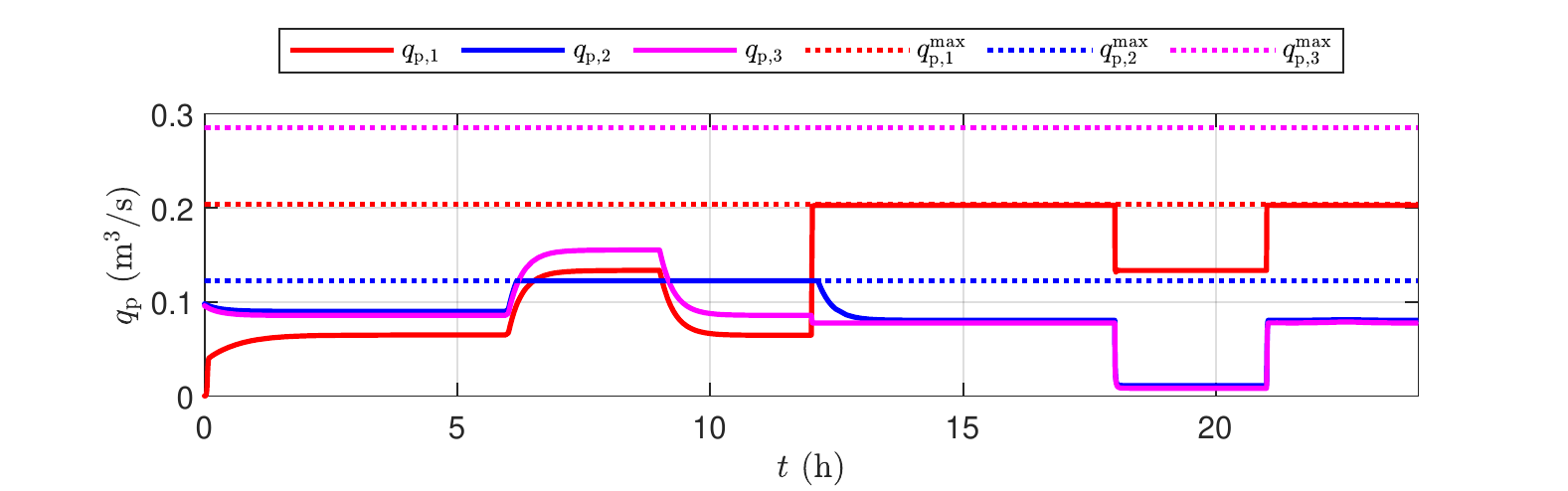}
\caption{~}
\end{subfigure}
\caption{{\bf Left column, top to bottom:} evolution of the producers' temperatures $T_{\mathrm{p},i}$, some of the consumers' outlet temperatures $T_{\mathrm{c},i}$ and the   volume of hot water in the storage tanks $V_{\mathrm{sh},i}$. {\bf Right column, top to bottom:} evolution of the inputs $P_{\mathrm{p},i}$, some of the inputs $q_{\mathrm{c},i}$, and the inputs $q_{\mathrm{p},i}$. The scenario is the same as for the results in Fig.~\ref{fig:sims_1}, but considering saturation to the control inputs.}
\label{fig:sims_sat_exp1}
\end{figure*}

\begin{figure*}
\centering
\begin{subfigure}{0.49\textwidth}
\includegraphics[width=\linewidth]{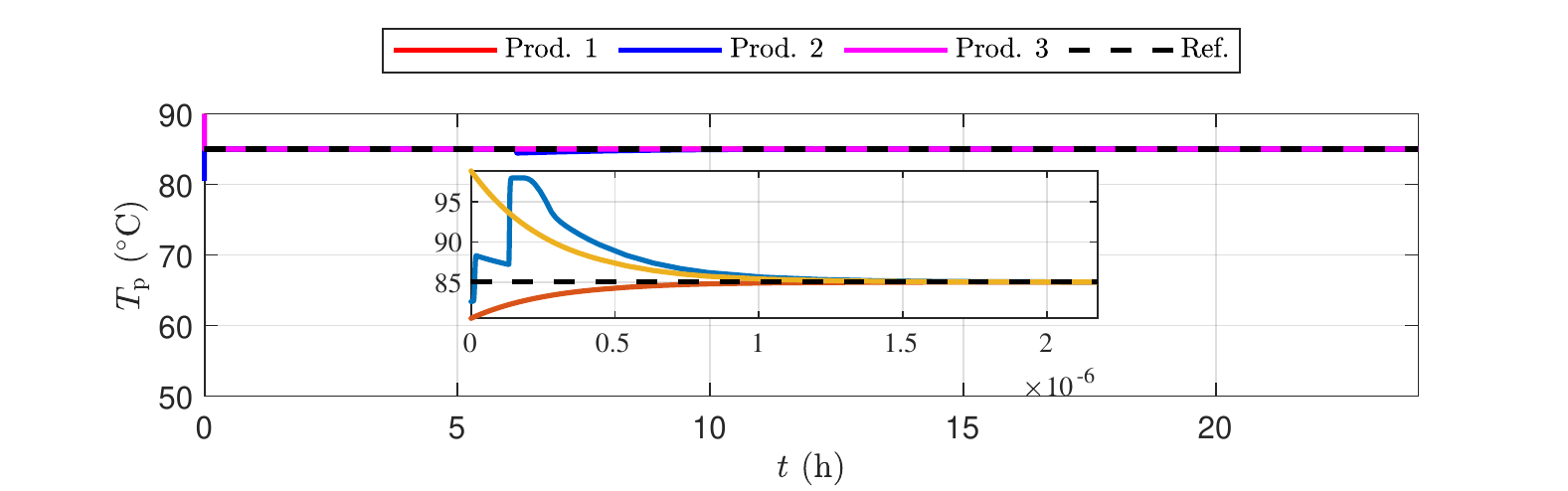}
\caption{~}
\end{subfigure}
\begin{subfigure}{0.49\textwidth}
\includegraphics[width=\linewidth]{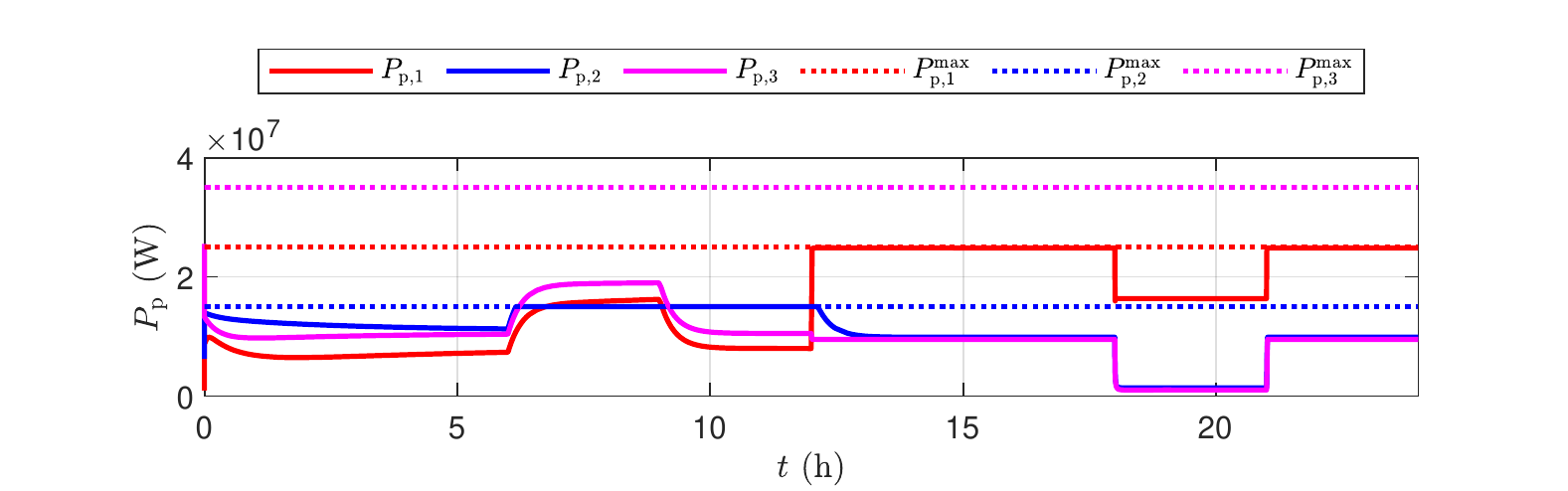}
\caption{~}
\end{subfigure}
\begin{subfigure}{0.49\textwidth}
\includegraphics[width=\linewidth]{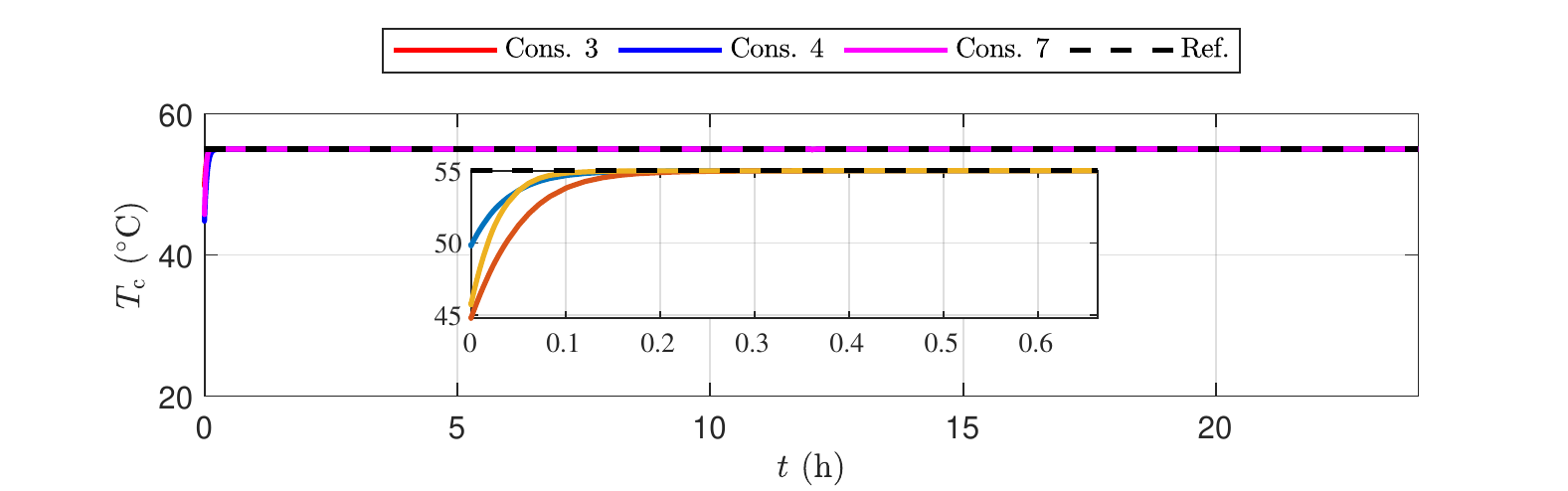}
\caption{~}
\end{subfigure}
\begin{subfigure}{0.49\textwidth}
\includegraphics[width=\linewidth]{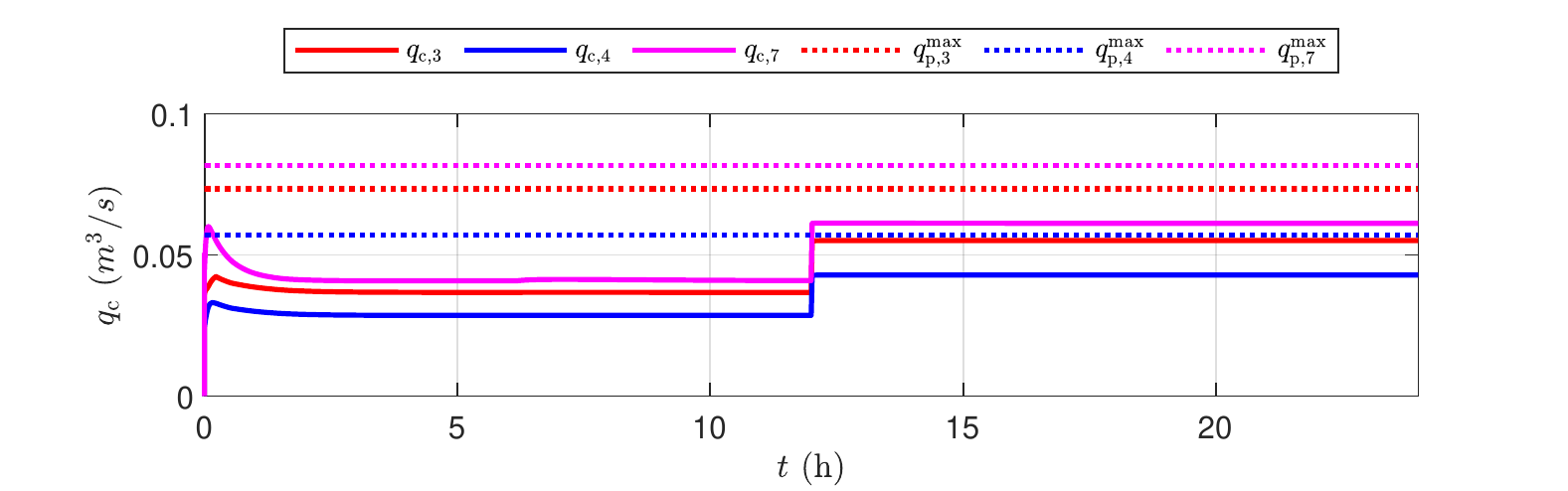}
\caption{~}
\end{subfigure}
\begin{subfigure}{0.49\textwidth}
\includegraphics[width=\linewidth]{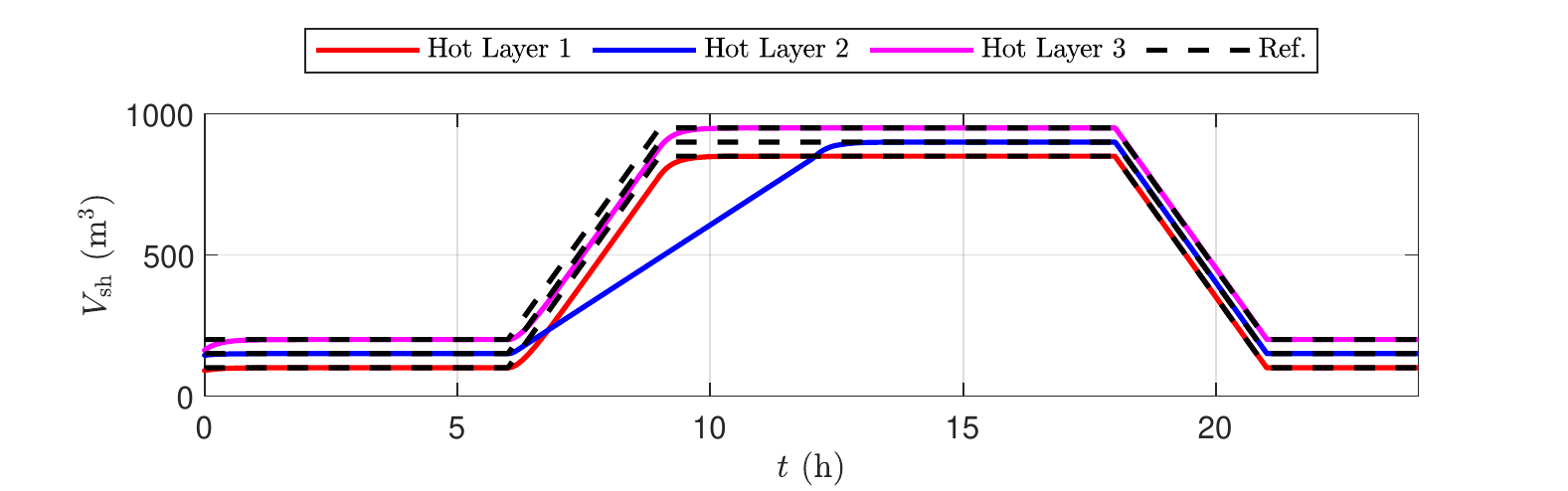}
\caption{~}
\end{subfigure}
\begin{subfigure}{0.49\textwidth}
\includegraphics[width=\linewidth]{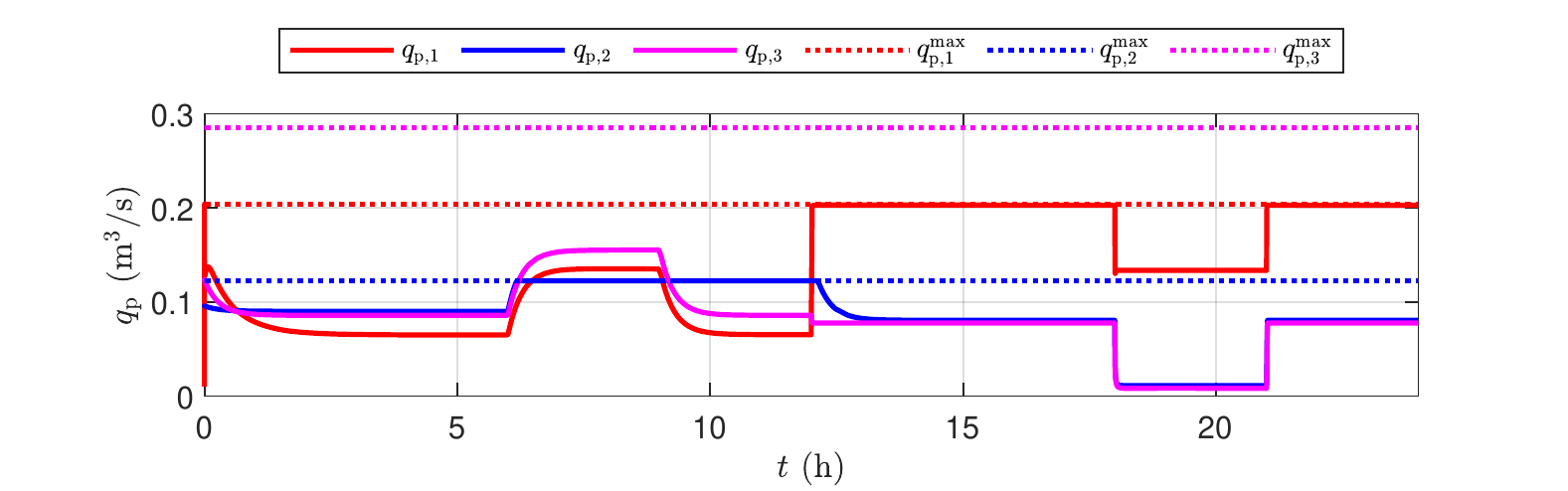}
\caption{~}
\end{subfigure}
\caption{Sample of additional numerical simulations we performed taking initial conditions that further deviate from a system equilibrium (up to 25\% deviation) and considering input constraints. {\bf Left column, top to bottom:} evolution of the producers' temperatures $T_{\mathrm{p},i}$, some of  consumers' outlet temperatures $T_{\mathrm{c},i}$ and the   volume of hot water in the storage tanks $V_{\mathrm{sh},i}$. {\bf Right column, top to bottom:} evolution of the inputs $P_{\mathrm{p},i}$, some of the inputs $q_{\mathrm{c},i}$, and the inputs $q_{\mathrm{p},i}$.}
\label{fig:sims_sat_exp2}
\end{figure*}

\section{Concluding remarks}
Done
In this letter, we have addressed  producer supply and consumer return temperature control in multi-producer DH systems through the design of novel decentralized controllers that also consider the regulation of the amount of hot  water of multiple, distributed storage tanks. The design is complemented with  a Lyapunov theory-based closed-loop stability analysis, from which convergence of the variables of interest is guaranteed. 
%The establishment of our results strongly relies on  the modularity of the developed (nonlinear) system model.  
{\color{black}Extensions to this work  we are currently investigating include:  time-varying heat demand profiles (c.f., \cite{Scholten_tcst_2015});   more detailed consumer models (c.f.,  \cite{Alisic2019}); fair energy distribution \cite{vandermeulen_control_review_18}; and input saturation \cite{Scholten2017}.}


%%%%%%%%%%%%%%%%%%%%%%%%%%%%%%%%%%%%%%%%%%%%%%%%%%%%%%%%%%%%%%%%%%%%%%%%%%%%%%%%

%%%%%%%%%%%%%%%%%%%%%%%%%%%%%%%%%%%%%%%%%%%%%%%%%%%%%%%%%%%%%%%%%%%%%%%%%%%%%%%%

%%%%%%%%%%%%%%%%%%%%%%%%%%%%%%%%%%%%%%%%%%%%%%%%%%%%%%%%%%%%%%%%%%%%%%%%%%%%%%%%
%\section*{APPENDIX}

%\AtNextBibliography{\footnotesize}
%\printbibliography 

%\balance
%\bibliographystyle{IEEEtran}
%% argument is your BibTeX string definitions and bibliography database(s)
%\bibliography{district_heating.bib}
%

\end{document}